\documentclass{article}
\usepackage{amssymb}

\title{ Two-valued states on Baer $^*$-semigroups }

\author{{\sc Hector Freytes}\thanks{%
Fellow of the Consejo Nacional de Investigaciones Cient\'{\i}ficas y
T\'ecnicas (CONICET)}\ \ $^{1,2}$, \  {\sc Graciela Domenech}$^{*}$
$^3$ \\ and \ {\sc Christian de Ronde}$^{*}$ $^{4,5}$}

\begin{document}

\bibliographystyle{plain}

\maketitle

\begin{center}

\begin{small}
1. Instituto Argentino de Matem\'atica (IAM) \\
Saavedra 15 - 3er piso - 1083 Buenos Aires, Argentina\\
e-mail: hfreytes@dm.uba.ar - hfreytes@gmail.com\\
2. Dipartimento di Matematica e Informatica ``U. Dini'' \\
Viale Morgagni, 67/a - 50134 Firenze, Italia\\
3. Instituto de Astronom\'{\i}a y F\'{\i}sica del Espacio (IAFE)\\
Casilla de Correo 67, Sucursal 28 . 1428 Buenos Aires, Argentina\\
e-mail: domenech@iafe.uba.ar\\
4. Instituto de Filosof\'{\i}a  ``Dr. Alejandro Korn''\\
(UBA-CONICET), Buenos Aires, Argentina \\
5. Center Leo Apostel (CLEA) and Foundations of  the Exact Sciences (FUND) \\
Brussels Free University,  Krijgskundestraat 33 - 1160 Brussels,
Belgium\\
e-mail: cderonde@vub.ac.be
\end{small}
\end{center}

\begin{abstract}
\noindent In this paper we develop an algebraic framework that
allows us to extend families of two-valued states on orthomodular
lattices to Baer $^*$-semigroups. We apply this general approach to
study the full class of two-valued states and the subclass of
Jauch-Piron two-valued states on Baer $^*$-semigroups.
\end{abstract}



\bibliography{pom}

\begin{thebibliography}{10}

\bibitem{AD} D. Adams,  ``Equational classes of Foulis semigroups and orthomodular lattices'', Proc. Univ. of Houston, Lattice Theory Conf.. Houston (1973) 486-497.

\bibitem{BvN} G. Birkhoff, and J. von Neumann,  ``The logic of quantum mechanics'', Ann. Math. {\bf 27}  (1936)
823-843.

\bibitem{BLIJ} T. S. Blyth, M. F. Janowitz,  ``Residuation Theory'', Pergamon Press, (1972)


\bibitem{Bur} S. Burris,  H. P.  Sankappanavar, {\em A Course in Universal
Algebra}, Graduate Text in Mathematics, Vol. 78. Springer-Verlag,
New York Heidelberg Berlin, 1981.

\bibitem{DFD} G. Domenech, H. Freytes,  C. de Ronde, ``Equational characterization for two-valued states in orthomodular quantum systems'', Rep. Math. Phys. {\bf 68}, (2011) 65-83.

\bibitem{DGG} M. L. Dalla Chiara , R. Giuntini,  R. Greechie, {\em Reasoning in Quantum Theory, Sharp and Unsharp Quantum Logics}, Kluwer, Dordrecht-Boston-London, 2004.


\bibitem{DV1} A. Dvure\u{c}enskij, ``On States on MV-algebras and their Applications'' J. Logic and Computation {\bf 21} (3), (2011) 407-427.


\bibitem{FOU} D. Foulis,  {\em Baer $^*$-semigroups}, Proccedings of American Mathematical Society {\bf 11}, (1960)
648-654.

\bibitem{F} D. Foulis, ``Observables, states, and symmetries in the context of $CB$-effect algebras'', Rep. Math. Phys. {\bf 60}, (2007) 329-346.

\bibitem{gleason} Gleason, A. M. ``Measures on the closed subspaces of a Hilbert space", Journal of Mathematics and Mechanics {\bf 6}, (1957) 885-893.


\bibitem{GUD} S. Gudder,  {\em Stochastic Methods in Quantum Mechanics}, Elseiver-North-Holand, New York 1979.


\bibitem{JAU1} J. Jauch, {\em Foundations of Quantum Mechanics}, Addison Wesley, Reading, Mass, 1968.


\bibitem{Ka} J. A. Kalman, ``Lattices with involution'', Trans. Amer. Math. Soc.  {\bf 87}, (1958) 485-491.


\bibitem{KAL}  G. Kalmbach,  {\em Ortomodular Lattices}, Academic Press, London, 1983.


\bibitem{KM} J. K\"{u}hr and D. Mundici, ``De Finetti theorem and Borel states in $[0,1]$-valued logic'', Int J. Approx. Reason {\bf 46} (3), (2007) 605-616.


\bibitem{KR} T. Kroupa, ``Every state on semisimple MV-algebra is integral'', Fuzzy Sets Syst. {\bf 157},  (2006) 2771-2787.


\bibitem{HP} J. Harding and P. Pt\'{a}k, ``On the set representation of an orthomodular poset'', Colloquium
Math. 89 (2001) 233-240


\bibitem{MAC} G. W. Mackey {\em Mathematical foundations of quantum mechanics} New York: W. A. Benjamin, Inc. 1963.

\bibitem{MES} A. Messiah  {\em Quantum mechanics}, Vol 1. Amsterdan: North-Holand Publishing Company 1961.

\bibitem{MM} F. Maeda and S. Maeda, {\em Theory of Symmetric
Lattices}, Springer-Verlag, Berlin, 1970.


\bibitem{NAV1} M. Navara, ``Descriptions of state spaces of orthomodular lattices'', Math. Bohemica {\bf 117}, (1992) 305-313.

\bibitem{NAV2} M. Navara, ``Triangular norms and measure of fuzzy set'', in {\em Logical, Algebraic, Analytic and Probabilistic
Aspect of Triangular Norms}, Elsevier, Amsterdan 2005, 345-390.

\bibitem{vNlibro} J. von Neumann, {\em Mathematical Foundations of Quantum
Mechanics}, Princeton University Press, 12th. edition, Princeton,
1996.


\bibitem{PIR1} C. Piron, {\em Foundations of Quantum Physics}, Benjamin, Reading, Mass. 1976.

\bibitem{POOL} J.C. Pool,  {\em Baer $^*$-semigroups and the logic of quantum mechanics}, Commun. Math. Phys. {\bf 9}, (1968) 118-141.

\bibitem{PUL} P. Pt\'{a}k, S. Pulmannov\'{a}, {\em Orthomodular Structures as Quantum Logics}, Kluwer, Dordrecht, 1991.

\bibitem{PUL1} S. Pulmannov\'{a}, ``Sharp and unsharp observables on s-MV algebras — A comparison with the Hilbert space approach'', Fuzzy Sets Syst. {\bf 159} (22), (2008) 3065-3077.

\bibitem{PTAK} P. Pt\'{a}k, ``Weak dispersion-free states and hidden variables hypothesis'', J. Math. Phys. {\bf 24} (1983) 839-840.

\bibitem{RIE} Z. Rie\u{c}anov\'{a},  ``Continuous lattice effect algebras admitting order-continuous states'', Fuzzy Sets Syst. {\bf 136}, (2003) 41-54.

\bibitem{RIE2} Z. Rie\u{c}anov\'{a}, ``Effect algebraic extensions of generalized effect algebras and two-valued
states'', Fuzzy Sets Syst. {\bf 159},   (2008)  1116-1122.

\bibitem{RU} G.T. R\"{u}ttimann ``Jauch-Piron states'', J. Math. Phys. {\bf 18} (1977) 189-193.

\bibitem{TK1} J. Tkadlec, ``Partially additive measures and set representation of orthoposets'', J. Pure Appl. Algebra  {\bf 86},  (1993) 79-94.

\bibitem{TK2} J. Tkadlec, ``Partially additive states on orthomodular posets'',  Colloquium Mathematicum {\bf LXII},  (1995) 7-14.

\bibitem{TK3} J. Tkadlec, ``Boolean orthoposets and two-valued states on them'', Rep. Math. Phys. {\bf 31}, (1992) 311-316.










\end{thebibliography}

\newtheorem{theo}{Theorem}[section]

\newtheorem{definition}[theo]{Definition}

\newtheorem{lem}[theo]{Lemma}

\newtheorem{met}[theo]{Method}

\newtheorem{prop}[theo]{Proposition}

\newtheorem{coro}[theo]{Corollary}

\newtheorem{exam}[theo]{Example}

\newtheorem{Problem}[theo]{Problem}

\newtheorem{rema}[theo]{Remark}{\hspace*{4mm}}

\newtheorem{example}[theo]{Example}

\newcommand{\proof}{\noindent {\em Proof:\/}{\hspace*{4mm}}}

\newcommand{\qed}{\hfill$\Box$}

\newcommand{\ninv}{\mathord{\sim}} %involutive negation

\section{Introduction}

Recently, several authors have paid attention to the study of the
concept of ``state'' by extending it to classes of algebras more
general than the  $\sigma$-algebras, as orthomodular posets
\cite{DGG, PUL}, MV-algebras \cite{DV1, KM, KR, NAV2, PUL1} or
effect algebras \cite{F, RIE, RIE2}. In the particular case of
quantum mechanics (QM), different families of states are
investigated not only because they provide different representations
of the event structure of quantum systems \cite{NAV1, TK1, TK2} but
also because of their importance in order to understand QM
\cite{GUD, JAU1, PIR1, PTAK}.

In \cite{DFD}, a general theoretical framework to study families of
two-valued states on orthomodular lattices is given. We shall use
these ideas for a general study of two-valued states extended to
Baer $^*$-semigroups. Moreover, we investigate varieties of Baer
$^*$-semigroups expanded with a unary operation that allows us to
capture the notion of two-valued states in an algebraic structure.

The paper is organized as follows: Section \ref{BASICNOTION}
contains generalities on universal algebra, orthomodular lattices,
and Baer $^*$-semigroups. In Section \ref{EVENTSEC}, motivations for
a natural extension of the concept of two-valued state from
orthomodular lattices to Baer $^*$-semigroups are presented. In
Section \ref{ALGEBRAICAP}, we introduce the concept of $IE_
B^*$-semigroup. It is presented as a Baer $^*$-semigroup with a
unary operation that enlarges the language of the structure. This
operation is defined by equations giving rise to a variety denoted
by ${\mathcal{IE}}^*_B$. In this way, ${\mathcal{IE}}^*_B$ defines a
common abstract framework in which several families of two-valued
states can be algebraically treated as unary operations on Baer
$^*$-semigroups.  In Section \ref{VARIETIES}, we give a decidable
procedure to extend equational theories of two-valued states on
orthomodular lattices to Baer $^*$-semigroups determining
sub-varieties of ${\mathcal{IE}}^*_B$. In Section \ref{FULLCLASS}
and Section \ref{JAUCHPIRON}, we apply the results obtained in an
abstract way to two important classes of two-valued states, namely
the full class of two-valued states and the subclass of Jauch-Piron
two-valued states. In Section \ref{APROBLEM}, we study some problems
about equational completeness related to subvarieties of
${\mathcal{IE}}^*_B$. Finally, in Section \ref{OPERATORE}, we
introduce subvarieties ${\mathcal{IE}}^*_B$ whose equational
theories are determined by classes of two-valued states on
orthomodular lattices.

\section{Basic notions} \label{BASICNOTION}

First we recall from \cite{Bur} some notions of universal algebra
that will play an important role in what follows. A {\em variety} is
a class of algebras of the same type defined by a set of equations.
If ${\mathcal A}$ is a variety and ${\mathcal B}$ is a subclass of
${\mathcal A}$, we denote by ${\mathcal V}({\mathcal B})$ the
subvariety of ${\mathcal A}$ generated by the class ${\mathcal B}$,
i.e. ${\mathcal V}({\mathcal B})$ is the smallest subvariety of
${\mathcal A}$ containing ${\mathcal B}$. Let ${\mathcal A}$ be a
variety of algebras of type $\tau$. We denote by Term$_{\mathcal A}$
the {\em absolutely free algebra} of type $\tau$ built  from the set
of variables $V = \{x_1, x_2,...\}$. Each element of Term$_{\mathcal
A}$ is referred as a {\em term}. We denote by Comp($t$) the
complexity of the term $t$ and by $t = s$ the equations of
Term$_{\mathcal A}$.

For $t\in$ Term$_{\mathcal A}$ we often write $t(x_1, \ldots x_n)$
to indicate that the variables occurring in $t$ are among $x_1,
\ldots x_n$. Let $A \in {\mathcal A}$. If $t(x_1, \ldots x_n) \in$
Term$_{\mathcal A}$ and $a_1,\dots, a_n \in A$, by $t^A(a_1,\dots,
a_n)$ we denote the result of the application of the term operation
$t^A$ to the elements $a_1,\dots, a_n$. A {\em valuation} in $A$ is
a function $v:V\rightarrow A$. Of course, any valuation $v$ in $A$
can be uniquely extended to an ${\mathcal A}$-homomorphism
$v:$Term$_{\mathcal A} \rightarrow A$ in the usual way, i.e., if
$t_1, \ldots, t_n \in$ Term$_{\mathcal A}$ then $v(t(t_1, \ldots,
t_n)) = t^A(v(t_1), \ldots, v(t_n))$. Thus, valuations are
identified with ${\mathcal A}$-homomorphisms from the absolutely
free algebra. If $t,s \in$ Term$_{\mathcal A}$, $A \models t = s$
means that for each valuation $v$ in $A$, $v(t) = v(s)$ and
${\mathcal A}\models t=s$ means that for each $A\in {\mathcal A}$,
$A \models t = s$.

For each algebra $A \in {\mathcal A}$, we denote by Con($A$) the
congruence lattice of $A$, the diagonal congruence is denoted by
$\Delta$ and the largest congruence $A^2$ is denoted by $\nabla$.
$\theta$ is called  {\em factor congruence} iff there is a
congruence $\theta^*$ on $A$ such that, $\theta \land \theta^* =
\Delta$, $\theta \lor \theta^* = \nabla$ and $\theta$ permutes with
$\theta^*$. If $\theta$ and $\theta^*$ is a pair of factor
congruences on $A$ then $A \cong A/\theta \times A/\theta^*$. $A$ is
{\em directly indecomposable} if $A$ is not isomorphic to a product
of two non trivial algebras or, equivalently, $\Delta,\nabla$ are
the only factor congruences in $A$. We say that $A$ is {\em
subdirect product} of a family of $(A_i)_{i\in I}$ of algebras if
there exists an embedding $f: A \rightarrow \prod_{i\in I} A_i$ such
that $\pi_i f : A\! \rightarrow A_i$ is a surjective homomorphism
for each $i\in I$ where $\pi_i$ is the projector onto $A_i$. $A$ is
{\em subdirectly irreducible} iff $A$ is trivial or there is a
minimum congruence in Con($A$)$ - \Delta$. It is clear that a
subdirectly irreducible algebra is directly indecomposable. An
important result due to Birkhoff is that every algebra $A$ is a
subdirect product of subdirectly irreducible algebras. Thus, the
class of subdirectly  irreducible algebras rules the valid equations
in the variety ${\mathcal A}$.

Now we recall from \cite{KAL, MM} some notions about orthomodular
lattices.  A {\em lattice with involution} \cite{Ka} is an algebra
$\langle L, \lor, \land, \neg \rangle$ such that $\langle L, \lor,
\land \rangle$ is a  lattice and $\neg$ is a unary operation on $L$
that fulfills the following conditions: $\neg \neg x = x$ and $\neg
(x \lor y) = \neg x \land \neg y$.  An {\em orthomodular lattice} is
an algebra $\langle L, \land, \lor, \neg, 0,1 \rangle$ of type
$\langle 2,2,1,0,0 \rangle$ that satisfies the following conditions:

\begin{enumerate}
\item
$\langle L, \land, \lor, \neg, 0,1 \rangle$ is a bounded lattice with involution,

\item
$x\land  \neg x = 0 $.

\item
$x\lor ( \neg x \land (x\lor y)) = x\lor y $

\end{enumerate}

We denote by ${\mathcal{OML}}$ the variety of orthomodular lattices.
Let $L$ be an orthomodular lattice. Two  elements $a,b$ in $L$ are
{\em orthogonal} (noted $a \bot b$) iff $a\leq \neg b$. For each
$a\in L$ let us consider the interval $[0,a] = \{x\in L : 0\leq x
\leq a \}$ and the unary operation in  $[0,a]$ given by $\neg_a x =
\neg x \land a$. As one can readily realize, the structure $L_a =
\langle [0,a], \land, \lor, \neg_a, 0, a \rangle$ is an orthomodular
lattice.

{\em Boolean algebras} are orthomodular lattices satisfying  the
{\em distributive law} $x\land (y \lor z) = (x \land y) \lor (x
\land z)$. We denote by ${\mathbf 2}$ the Boolean algebra of two
elements. Let $L$ be an orthomodular lattice. An element $c\in L$ is
said to be a {\em complement} of $a$ iff $a\land c = 0$ and $a\lor c
= 1$. Given $a, b, c$ in $L$, we write: $(a,b,c)D$\ \   iff $(a\lor
b)\land c = (a\land c)\lor (b\land c)$; $(a,b,c)D^{*}$ iff $(a\land
b)\lor c = (a\lor c)\land (b\lor c)$ and $(a,b,c)T$\ \ iff
$(a,b,c)D$, (a,b,c)$D^{*}$ hold for all permutations of $a, b, c$.
An element $z$ of $L$ is called {\em central} iff for all elements
$a,b\in L$ we have\ $(a,b,z)T$. We denote by $Z(L)$ the set of all
central elements of $L$ and it is called the {\em center} of $L$.

\begin{prop}\label{eqcentro} Let $L$ be an orthomodular lattice. Then we have:

\begin{enumerate}

\item
$Z(L)$ is a Boolean sublattice of $L$ {\rm \cite[Theorem 4.15]{MM}}.

\item
$z \in Z(L)$ iff for each $a\in L$, $a = (a\land z) \lor (a \land \neg z)$  {\rm \cite[Lemma 29.9]{MM}}.
\end{enumerate}
\rule{5pt}{5pt}
\end{prop}

Now we recall from \cite{AD, FOU, KAL}  some notions about Baer
$^*$-semigroups. A {\em Baer $^*$-semigroup} \cite{FOU} also called
Foulis semigroup \cite{AD, BLIJ, KAL} is an algebra $\langle S,
\cdot , ^*, ', 0 \rangle$ of type $\langle 2,1,1,0 \rangle$ such
that, upon defining $1= 0'$, the following conditions are satisfied:

\begin{enumerate}
\item
$\langle S, \cdot \rangle$ is a semigroup,

\item
$0\cdot x = x \cdot 0 = 0$,

\item
$1\cdot x = x \cdot 1 = x$,

\item
$(x \cdot y)^* = y^* \cdot x^*$,

\item
$x^{**} = x $,

\item
$x\cdot x' = 0$,

\item
$x' \cdot x' = x' = (x')^*$,

\item
$x'\cdot y \cdot (x\cdot y)' = y\cdot (x\cdot y)' $.

\end{enumerate}

Let $S$ be a Baer $^*$-semigroup. An element $e\in S$ is a {\em
projector} iff $e = e^* = e\cdot e$. The set of all projectors of
$S$ is denoted by $P(S)$. A projector $e \in P(S)$ is said to be
closed iff $e'' = e$. We denote by $P_c(S)$ the set of all closed
projectors. Moreover we can prove that: $$P_c(S) = \{x': x\in S \}$$
We can define a partial order $\langle P(S), \leq \rangle$ as
follows: $$e \leq f \Longleftrightarrow e \cdot f = e$$

In {\rm \cite[Theorem 37.2]{MM}} it is proved that, for any $e, f
\in P_c(S)$, $e \leq f$  iff $e\cdot S \subseteq f\cdot S$. The
facts stated in the next proposition are either proved in \cite{FOU}
or follow immediately from the results in \cite{FOU}:

\begin{prop}\label{PBAER1}
Let $S$ be a Baer $^*$-semigroup. Then:

\begin{enumerate}
\item
If $x,y \in P(S)$ and $x\leq y$ then $y'\leq x'$,

\item
$(x\cdot y)'' = (x'' \cdot y)'' \leq y''$,

\item
$(x^* \cdot x)'' = x''$,

\item
for each $x\in P_c(S)$, $0\leq x \leq 1$,

\item
$x \cdot y = 0$ \hspace{0.2cm} iff \hspace{0.2cm} $y = x'\cdot y$

\end{enumerate}
\rule{5pt}{5pt}
\end{prop}

Observe that, item 5  was one of the original conditions in the
definition of a Baer *-semigroup in \cite{FOU}. In the presence of
conditions 1 ...7 of the definition of a Baer
*-semigroup, the latter condition is equivalent to condition 8 (see \cite[Proposition 2]{AD}).

\begin{theo}\label{PRO1}{\rm \cite[Theorem 37.8]{MM}}
Let $S$ be a Baer $^*$-semigroup. For any $e_1, e_2 \in P_c(S)$,
we define the following operations:

\begin{enumerate}
\item[]
$e_1 \land e_2 = e_1\cdot (e_2' \cdot e_1)'$,

\item[]
$e_1 \lor e_2 = (e_1' \land e_2')'$.

\end{enumerate}

\noindent Then $\langle P_c(S), \land, \lor, ', 0,1  \rangle$ is an
orthomodular lattice with respect to the order $\langle P(S), \leq
\rangle$. \rule{5pt}{5pt}
\end{theo}

We can build a Baer $^*$-semigroup from an orthomodular lattice \cite{FOU}. In the following we briefly describe this construction.

Let $\langle A, \leq, 0,1 \rangle$ be a bounded partial ordered set.
An order-preserving function $\phi: A \rightarrow A$ is called {\em
residuated function} iff there is another order-preserving function
$\phi^+: A \rightarrow A$, called a {\em residual function} of
$\phi$ such that $\phi \phi^+(x)  \leq x \leq  \phi^+\phi (x)$. It
can be proved that if $\phi$ admits a residual function $\phi^+$, it
is completely determined  by $\phi$.

\begin{rema}
{\rm We will adopt the notation in \cite[$\S$1]{AD} in which
residuated functions are written on the right. More precisely, if
$\phi, \psi$ are residuated functions, $x\phi$ indicates the value
$\phi(x)$ and $\psi \phi$ is interpreted as the function $x\psi \phi
= (x\psi) \phi$. }
\end{rema}

We denote by $S(A)$ the set of residuated functions of $A$. Let $\theta$ be the constant
function in $A$ given by $x\theta = 0$. Clearly $\theta$ is an order-preserving function
and $\theta^+$ is the constant function $x\theta^+ = 1$. Thus $\theta \in S(A)$ and $\langle S(A),
\circ, \theta  \rangle$, where $\psi \circ \phi =  \psi \phi$, is a semigroup.

\begin{theo}\label{PRO2}{\rm \cite[Proposition 2]{AD}}
Let $L$ be an orthomodular lattice. For each $a\in L$ we define
\begin{center}
$x\phi_a = (x\lor \neg a) \land a$ {\rm({\em Sasaki projection})}
\end{center}
If we define the following unary operations in $S(L)$:

\begin{itemize}
\item[]
$\phi^*$: such that $x\phi^* = \neg ((\neg x) \phi^+ )$,

\item[]
$\phi': = \phi_{\neg 1\phi} $

\end{itemize}

then:

\begin{enumerate}
\item
$\langle S(L), \circ, ^*, ', \theta \rangle$ is a Baer $^*$-semigroup.

\item
$P_c(S(L)) = \{\phi_a : a \in L \}$,

\item
$f_L:L \rightarrow P_c(S(L))$ such that $f_L(a) = \phi_a$ is an
${\mathcal{OML}}$-isomorphism.
\end{enumerate}
\rule{5pt}{5pt}
\end{theo}

If $L$ is an orthomodular lattice, the Baer $^*$-semigroup $\langle
S(L), \circ, ^*, ', \theta \rangle$, or $S(L)$ for short, will be
referred to  as {\em the  Baer $^*$-semigroup of the residuated
functions of $L$}.

Let $L$ be an orthomodular lattice. We say that a Baer
$^*$-semigroup $S$ {\em coordinatizes} $L$ iff $L$ is
${\mathcal{OML}}$-isomorphic to $P_c(S)$.

\section{Two-valued states and Baer $^*$-semigroups } \label{EVENTSEC}

The study of  two-valued states becomes relevant in different
frameworks. From a physical point of view, two-valued states are
distinguished among the set of all classes of states because of
their relation to hidden variable theories of quantum mechanics
\cite{GUD}. Another motivation for the analysis of  two-valued
states is rooted in the study of algebraic and topological
representations of the event structures in quantum logic. Examples
of them are the characterization of Boolean orthoposets by means of
two-valued states \cite{TK3} and the representation of orthomodular
lattices via clopen sets in a compact Hausdorff closure space
\cite{TK2}, later  extended to orthomodular posets in \cite{HP}. We
are interested in a theory of two-valued states on Baer
$^*$-semigroups as a natural extension of two-valued states on
orthomodular lattice. Formally, a {\em two-valued state} on an
orthomodular $L$ is a function $\sigma:L \rightarrow \{0,1\}$
satisfying the following :
\begin{enumerate}
\item
$\sigma(1) = 1$,

\item
if $x \bot y$ then $\sigma(x \lor y) = \sigma(x) + \sigma(y)$.
\end{enumerate}

Let $L$ be an orthomodular lattice and $\sigma:L \rightarrow
\{0,1\}$ be a two-valued state. The following properties are derived
directly from the definition of two-valued state:

\begin{center}
$\sigma(\neg x) = 1-\sigma(x)$ \hspace{0.2cm} and \hspace{0.2cm} if $x\leq y$ then $\sigma(x) \leq \sigma(y)$
\end{center}

Based on the above mentioned two properties, in \cite{DFD}, Boolean
pre-states are introduced as a general theoretical framework to
study families of two-valued states on orthomodular lattices. We
shall use these ideas for a general study of two-valued states
extended to Baer $^*$-semigroups. Thus, we first give the definition
of Boolean pre-state.

\begin{definition}
{\rm Let $L$ be an orthomodular lattice. By a {\em Boolean
pre-state} on $L$ we mean a function $\sigma:L \rightarrow \{0,1\}$
such that:
\begin{enumerate}
\item
$\sigma(\neg x) = 1 - \sigma(x)$,
\item
if $x\leq y$ then $\sigma(x) \leq \sigma(y)$.
\end{enumerate}
}
\end{definition}

\begin{example}
{\rm Let us consider the orthomodular lattice $MO2 \times {\mathbf
2}$ whose Hasse diagram has the following form: \vspace{0.8cm}

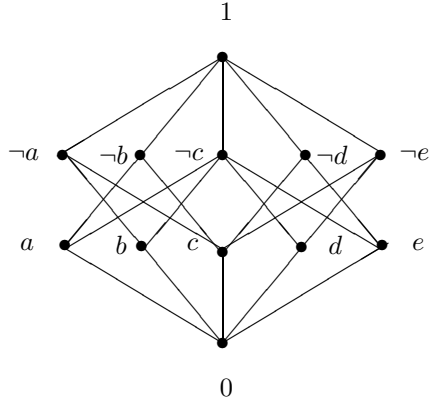
\begin{center}
\unitlength=1mm
\begin{picture}(60,20)(0,0)
\put(18,20){\line(5,-3){21}}
\put(18,20){\line(5,-6){21}}
\put(18,20){\line(0,-8){13}}
\put(18,20){\line(-5,-6){21}}
\put(18,20){\line(-5,-3){22}}

\put(18,-18){\line(5,6){21}}
\put(18,-18){\line(-5,6){21}}
\put(18,-18){\line(0,6){12}}
\put(18,-18){\line(-5,3){21}}
\put(18,-18){\line(5,3){22}}

\put(18,20){\circle*{1.5}}
\put(-3.3,7){\circle*{1.5}}
\put(7,7){\circle*{1.5}}
\put(29,7){\circle*{1.5}}
\put(39,7){\circle*{1.5}}
\put(18,-18){\circle*{1.5}}
\put(18,7){\circle*{1.5}}

\put(-3,-5){\circle*{1.5}}
\put(7.2,-5.2){\circle*{1.5}}
\put(28.5,-5.3){\circle*{1.5}}
\put(39.2,-5){\circle*{1.5}}
\put(18,-6){\circle*{1.5}}

\put(21,26){\makebox(-5,0){$1$}}
\put(-6,7){\makebox(-5,0){$\neg a$}}
\put(6,7){\makebox(-5,0){$\neg b$}}
\put(16,7){\makebox(-5,0){$\neg c$}}
\put(35,7){\makebox(-5,0){$\neg d$}}
\put(46,7){\makebox(-5,0){$\neg e$}}
\put(21,-24){\makebox(-5,0){$0$}}

\put(-6,-5){\makebox(-4,0){$ a$}}
\put(6,-5){\makebox(-3,0){$ b$}}
\put(16,-5){\makebox(-4,0){$ c$}}
\put(35,-5){\makebox(-4,0){$ d$}}
\put(46,-5){\makebox(-4,0){$ e$}}

\put(7,7){\line(5,-6){11}}
\put(18,7){\line(5,-6){10.5}}
\put(18,7){\line(-5,-6){10.5}}
\put(29,7){\line(-5,-6){10.5}}

\put(18,7){\line(-5,-3){21}}
\put(39,7){\line(-5,-3){21}}
\put(-3,7){\line(5,-3){21}}
\put(18,7){\line(5,-3){21}}

\end{picture}
\end{center}

\vspace{2.5cm}

\noindent If we define the function $\sigma: MO2 \times {\mathbf 2}
\rightarrow \{0,1\}$ such that: $$\sigma(x) = \cases {1, & if $x \in
\{1,\neg a, \neg b, \neg c, \neg d, \neg e\}$ \cr 0 , & if $x \in
\{0, a, b, c, d, e\}$ \cr}$$ we can see that $\sigma$ is a Boolean
pre-state. This function fails to be a two-valued states since, $b
\leq \neg c$ but $\sigma(b\lor c) \not = \sigma(b) + \sigma(c)$. In
fact $\sigma(b\lor c) = \sigma(\neg a) = 1$ and $\sigma(b) +
\sigma(c) = 0$.}
\end{example}

We denote by  ${\mathcal E}_B$ the category whose objects are pairs
$(L,\sigma)$ such that $L$ is an orthomodular lattice and $\sigma$
is a  Boolean pre-state on $L$.  Arrows in ${\mathcal E}_B$ are
$(L_1, \sigma_1) \stackrel{f}{\rightarrow} (L_2, \sigma_2) $ such
that $f:L_1 \rightarrow L_2$ is an $OML$-homomorphism, and the
following diagram is commutative:

\begin{center}
\unitlength=1mm
\begin{picture}(20,20)(0,0)
\put(8,16){\vector(3,0){5}} \put(2,10){\vector(0,-2){5}} \put(10,4){\vector(1,1){7}}

\put(2,10){\makebox(13,0){$\equiv$}}

\put(1,16){\makebox(0,0){$L_1$}} \put(20,16){\makebox(0,0){$\{0,1\}$}}
 \put(2,0){\makebox(0,0){$L_2$}}
 \put(2,20){\makebox(17,0){$\sigma_1$}}
 \put(2,8){\makebox(-6,0){$f$}}
\put(18,2){\makebox(-4,3){$\sigma_2$}}
\end{picture}
\end{center}

\noindent
These arrows are called ${\mathcal E}_B$-homomorphisms.\\

Let $L$ be an orthomodular lattice and let $\sigma:L \rightarrow
\{0,1\}$ be a Boolean pre-state. Since $L$ we can identify with
$P_c(S(L))$, we ask whether the Boolean pre-state $\sigma$ admits a
natural extension to the whole of  $S(L)$. In other words, whether
there exists some kind of function of the form $\sigma^*:S(L)
\rightarrow \{0,1\}$ such that the following diagram is commutative:

\begin{center}
\unitlength=1mm
\begin{picture}(20,20)(0,0)
\put(8,16){\vector(3,0){5}} \put(2,10){\vector(0,-2){5}} \put(10,4){\vector(1,1){7}}

\put(2,10){\makebox(13,0){$\equiv$}}

\put(1,16){\makebox(0,0){$L$}} \put(20,16){\makebox(0,0){$\{0,1\}$}}
 \put(2,0){\makebox(0,0){$S(L)$}}
 \put(2,20){\makebox(17,0){$\sigma$}}
 \put(2,8){\makebox(-6,0){$f_L$}}
\put(18,2){\makebox(-4,3){$\sigma^*$}}
\end{picture}
\end{center}

\noindent where $f_L$ is the ${\mathcal{OML}}$-isomorphism  $f_L: L
\rightarrow P_c(S(L))$ given in Theorem \ref{PRO2}-3. The simplest
way to do this would be to associate with each element $\phi \in
S(L)$ an appropriate closed projection $\phi_x \in P_c(S(L))$ and to
define $\sigma^*(\phi) = \sigma^*(\phi_x) = \sigma(x)$. An obvious
choice for $\phi_x$ is $\phi'' = \phi_{1\phi}$. In virtue of this
suggestion, we introduce the following concept:

\begin{definition}\label{BOOLEANSTAR}
{\rm Let $S$ be a Baer $^*$-semigroup. A Boolean$^*$ pre-state over
$S$ is a function $\sigma: S \rightarrow \{0,1\}$ such that

\begin{enumerate}
\item
$\sigma(x') = 1 - \sigma(x)$

\item
the restriction $\sigma/_{P_c(S)}$ is a Boolean pre-state on $P_c(S)$.

\end{enumerate}
}
\end{definition}

We denote by ${\mathcal E}^*_B$ the category whose objects are pairs
$(S,\sigma)$ such that $S$ is a Baer $^*$-semigroup and $\sigma$ is
a Boolean pre-state on $S$. Arrows in ${\mathcal E}^*_B$ are $(S_1,
\sigma_1) \stackrel{f}{\rightarrow} (S_2, \sigma_2) $ such that
$f:S_1 \rightarrow S_2$ is a Baer $^*$-semigroup homomorphism, and
the following diagram is commutative:

\begin{center}
\unitlength=1mm
\begin{picture}(20,20)(0,0)
\put(8,16){\vector(3,0){5}} \put(2,10){\vector(0,-2){5}} \put(10,4){\vector(1,1){7}}

\put(2,10){\makebox(13,0){$\equiv$}}

\put(1,16){\makebox(0,0){$S_1$}} \put(20,16){\makebox(0,0){$\{0,1\}$}}
 \put(2,0){\makebox(0,0){$S_2$}}
 \put(2,20){\makebox(17,0){$\sigma_1$}}
 \put(2,8){\makebox(-6,0){$f$}}
\put(18,2){\makebox(-4,3){$\sigma_2$}}
\end{picture}
\end{center}

These arrows are called ${\mathcal E}^*_B$-homomorphisms. Up to now
we have presented a notion that would naturally extend  the notion
of Boolean pre-state to Baer $^*$-semigroups. However we have not
yet proved that this extension may be  formally  realized. This will
be shown in Theorem \ref{BSIGMA3}. To see this, we first need the
following basic results:

\begin{prop}\label{BS1}
Let $S$ be a Baer $^*$-semigroup and $\sigma$ be a Boolean$^*$ pre-state on $S$. Then

\begin{enumerate}

\item
$\sigma(x'') = \sigma(x)$.

\item
If $x,y \in P(S)$ and $x\leq y$ then, $\sigma(y')\leq \sigma(x')$ and  $\sigma(x)\leq \sigma(y)$.

\item
$\sigma(x \cdot y) = \sigma(x'' \cdot y) \leq \sigma(y)$.

\item
$\sigma(x^* \cdot x ) = \sigma(x)$.

\end{enumerate}
\end{prop}

\begin{proof}
1) Is immediate. 2) Suppose that $x,y \in P(S)$ and $x\leq y$.  By
Proposition \ref{PBAER1}-1, $y'\leq x'$ and taking into account that
$x',y' \in P_c(S)$, $\sigma(y')\leq \sigma(x')$.   By Proposition
\ref{PBAER1}-1 again and since $y'\leq x'$ we have that $x'' \leq
y''$. Hence, by item 1, $\sigma(x) = \sigma(x'') \leq  \sigma(y'') =
\sigma(y)$. 3) By Proposition \ref{PBAER1}-2, $(x\cdot y)'' = (x''
\cdot y)'' \leq y''$. By item 1, $\sigma(x\cdot y) = \sigma((x\cdot
y)'') = \sigma((x'' \cdot y)'') = \sigma(x'' \cdot y)$. Since  $(x''
\cdot y)''$ and  $y''$ are closed projections, by item 1, we have
that $\sigma(x'' \cdot y) =\sigma((x'' \cdot y)'') \leq \sigma(y'')
= \sigma(y)$. 4) By Proposition \ref{PBAER1}-3  $(x^* \cdot x)'' =
x''$. Then, by item 1, $\sigma(x^* \cdot x) = \sigma((x^* \cdot
x)'') = \sigma(x'') = \sigma(x)$. \rule{5pt}{5pt}
\end{proof}

\begin{theo}\label{BSIGMA3}
Let $S$ be a Baer $^*$-semigroup and $\sigma$ a Boolean pre-state on
$P_c(S)$. Then $\sigma_S$ defined as  $$\sigma_S(x) = \sigma(x'')$$
is the unique  Boolean$^*$ pre-state on $S$ such that $\sigma_S
/_{P_c(S)} = \sigma $.

\end{theo}

\begin{proof}
If $x\in S$ then $x'' \in P_c(S)$ and $\sigma(x'')$ is defined. Then
$\sigma_S$ is well defined as a function. Note that if $x \in
P_c(S)$ then $\sigma_S(x) = \sigma(x'') = \sigma(x)$ since $'$ is an
orthocomplementation on the orthomodular lattice $P_c(S)$. Thus
$\sigma_S /_{P_c(S)} = \sigma$.  Let $x\in S$. Then $\sigma_S(x') =
\sigma(x''') = 1- \sigma(x'') = 1- \sigma_S(x)$. Thus $\sigma_S$ is
a  Boolean$^*$ pre-state on $S$. Let $\sigma_1$ be a Boolean$^*$
pre-state on $S$ such that $\sigma_1 /_{P_c(S)} = \sigma $. Let
$x\in S$. Since $x''\in P_c(S)$, by Proposition \ref{BS1}-2,
$\sigma_1(x) = \sigma_1(x'') = \sigma(x'') = \sigma_S(x)$. Hence
$\sigma_1 = \sigma_S$ and $\sigma_S$ is the unique  Boolean$^*$
pre-state on $S$ such that $\sigma_S /_{P_c(S)} = \sigma $.
\rule{5pt}{5pt}
\end{proof}

\section{An algebraic approach for two-valued states on Baer $^*$-semigroups}\label{ALGEBRAICAP}

In this section we study a variety of  Baer $^*$-semigroups enriched
with a unary operation that allows us to capture the concept of
two-valued states on Baer $^*$-semigroups in an equational theory. We
begin this section showing a way to deal with families of Boolean
pre-states on orthomodular lattices as varieties in which the
concept of Boolean pre-states is captured by adding a unary
operation to the orthomodular lattice structure.

Let $L$ be an orthomodular lattice and $\sigma: L \rightarrow
\{0,1\}$ be a Boolean pre-state. If we define the function $s:L
\rightarrow Z(L)$ such that $s(x) = 0^L$ if $\sigma(x) = 0$ and
$s(x) = 1^L$ if $\sigma(x) = 1$. Then $s$ has properties s1...s5 in
the following definition:

\begin{definition}\label{E}
{\rm An {\em orthomodular lattice with internal Boolean pre-state}
{$IE_B$-lattice for short}  is an algebra $ \langle L, \land, \lor,
\neg, s, 0, 1 \rangle$ of type $ \langle 2, 2, 1,1, 0, 0 \rangle$
such that $ \langle L, \land, \lor, \neg, 0, 1 \rangle$ is an
orthomodular lattice and $s$ satisfies the following equations for
each $x,y \in L$:

\begin{enumerate}

\item[\rm{s1.}]
$s(1) = 1$.

\item[\rm{s2.}]
$s(\neg x) = \neg s(x)$,

\item[\rm{s3.}]
$s(x \lor s(y)) = s(x) \lor s(y)$,

\item[\rm{s4.}]
$y = (y \land s(x)) \lor (y \land \neg s(x)) $,

\item[\rm{s5.}]
$s(x \land y) \leq s(x)\land s(y) $.

\end{enumerate}
}

\end{definition}

\noindent Thus, the class of $IE_B$-lattices is a variety that we
call ${\mathcal {IE}}_B$. The following proposition provides the
main properties of $IE_B$-lattices.

\begin{prop}\label{E1} {\rm \cite[Proposition 3.5]{DFD}}
Let $L$ be a $IE_B$-lattice. Then we have:

\begin{enumerate}
\item
$\langle s(L), \lor, \land, \neg, 0, 1 \rangle$ is a Boolean sublattice of $Z(L)$,

\item
If $x\leq y$ then $s(x) \leq s(y)$,

\item
$s(x) \lor s(y)  \leq  s(x\lor y)$,

\item
$s(s(x)) = s(x)$,

\item
$x\in s(L)$ iff $s(x) = x$,

\item
$s(x\land s(y))= s(x)\land s(y)$.
\rule{5pt}{5pt}
\end{enumerate}
\end{prop}

A crucial question that must be answered is under which conditions a
class of two-valued states over an orthomodular lattice can be
characterized by a subvariety of ${\mathcal{IE}}_B$. To do this, we
first need the following two basic results:

\begin{prop}\label{FUNC00}{\rm \cite[Theorem 4.4]{DFD}}
Let $L$ be an $IE_B$-lattice. Then there exists a Boolean pre-state
$\sigma:L \rightarrow \{0,1\}$ such that $\sigma(x) = 1$ iff
$\sigma(s(x))=1$. \rule{5pt}{5pt}
\end{prop}

Observe that, the Boolean pre-state in the last proposition is not
necessarily unique. When we have an $IE_B$-lattice and a Boolean
pre-state $\sigma:L \rightarrow \{0,1\}$ such that $\sigma(x) = 1$
iff $\sigma(s(x))=1$, we say that $s, \sigma$ are {\em coherent}. On
the other hand, we can build $IE_B$-lattices from objects in the
category ${\mathcal E}_B$ as shown the following proposition:

\begin{prop}\label{FUNC0}{\rm \cite[Theorem 4.10]{DFD}}
Let $L$ be an orthomodular lattice and $\sigma$ be a Boolean
pre-state on $L$. If we define ${\mathcal I}(L) = \langle L, \land,
\lor, \neg,  s_{\sigma}, 0,1  \rangle$ where $$ s_{\sigma}(x) =
\cases {1^L, & if $\sigma(x)=1$ \cr 0^L , & if $\sigma(x)=0$ \cr} $$
then:
\begin{enumerate}
\item
${\mathcal I}(L)$ is a $IE_B$-lattice and $s_{\sigma}$ is coherent
with $\sigma$.

\item
If $(L_1, \sigma_1) \stackrel{f}{\rightarrow} (L_2, \sigma_2) $ is a
${\mathcal E}_B$-homomorphism then $f:{\mathcal I}(L_1) \rightarrow
{\mathcal I}(L_2)$ is a $IE_B$-homomorphism.
\end{enumerate}
\rule{5pt}{5pt}
\end{prop}

Note that, ${\mathcal I}$ in the above proposition defines a functor
of the form  ${\mathcal I}: {\mathcal E}_B \rightarrow
{\mathcal{IE}}_B $. Now it is very important to characterize the
class $\{{\mathcal I}(L): L \in {\mathcal{IE}}_B \}$. To do this,
directly indecomposable algebras in ${\mathcal{IE}}_B$ play an
important role and the following proposition provides this result:

\begin{prop}\label{PROD2} {\rm \cite[Proposition 5.6]{DFD}}
Let $L$ be an $IE_B$-lattice. then:

\begin{enumerate}
\item
$L$ is directly indecomposable in ${\mathcal{IE}}_B$ iff $s(L) =
{\mathbf 2}$.

\item
If $L$ is directly indecomposable in ${\mathcal{IE}}_B$  then the
function $$ \sigma_s(x) = \cases {1, & if $s(x)=1^L$ \cr 0 , & if
$s(x)=0^L$ \cr} $$ is the unique Boolean pre-state coherent with
$s$.
\end{enumerate}
\rule{5pt}{5pt}
\end{prop}

\noindent Thus, an immediate consequence of Proposition \ref{FUNC0} and Proposition \ref{PROD2} is the following proposition:

\begin{prop}\label{PROD3}
Let ${\mathcal D}({\mathcal{IE}}_B)$ be the class of directly
indecomposable algebras in ${\mathcal{IE}}_B$. Then  $${\mathcal
D}({\mathcal{IE}}_B) = \{{\mathcal I}(L): L \in {\mathcal{IE}}_B
\}$$ and ${\mathcal I}: {\mathcal E}_B \rightarrow {\mathcal
D}({\mathcal{IE}}_B) $ is a categorical equivalence when we consider
${\mathcal D}({\mathcal{IE}}_B)$ as a category whose arrows are
$IE_B$-homomorphisms. \rule{5pt}{5pt}
\end{prop}

Since ${\mathcal D}({\mathcal{IE}}_B)$ contains the subdirectly
irreducible algebras of ${\mathcal{IE}}_B$, we have that:
$${\mathcal{IE}}_B \models t= s \hspace{0.4cm}  \mathrm{iff}
\hspace{0.4cm} {\mathcal D}({\mathcal{IE}}_B) \models t= s$$ Hence,
the class of orthomodular lattices admitting Boolean pre-states can
be identified with the directly indecomposable algebras in
${\mathcal{IE}}_B$ that determine the variety ${\mathcal{IE}}_B$. We
can use these ideas to give a general criterium to characterize
families of two-valued states over orthomodular lattices by a
subvariety of ${\mathcal{IE}}_B$.

Let ${\mathcal A}_I$ be a subvariety of ${\mathcal{IE}}_B$. We
denote by ${\mathcal D}({\mathcal A}_I)$ the class of directly
indecomposable algebras in ${\mathcal A}_I$.

\begin{definition} \label{DEF1}
{\rm Let ${\mathcal A}$ be a subclass of ${\mathcal E}_B$ and let
${\mathcal A}_I$ be a subvariety of ${\mathcal{IE}}_B$. Then we say
that ${\mathcal A}_I$ equationally characterizes ${\mathcal A}$ iff
the following two conditions are satisfied

\begin{enumerate}

\item[{\bf I:}]
For each $(L, \sigma) \in {\mathcal A}$, $\langle {\mathcal I}(L),
\land, \lor, \neg,  s_{\sigma}, 0,1 \rangle$ belong to ${\mathcal
D}({\mathcal A}_I)$ where $ s_{\sigma}(x) = \cases {1^L, & if
$\sigma(x)=1$ \cr 0^L, & if $\sigma(x)=0$ \cr} $

\item[{\bf E:}]
For  each $L \in {\mathcal D}({\mathcal A}_I)$, $(L, \sigma_s ) \in
{\mathcal A}$ where $\sigma_s$, the unique Boolean pre-state
coherent with $s$, is given by $ \sigma_s(x) = \cases {1, & if
$s(x)=1^L$ \cr 0 , & if $s(x)=0^L$ \cr} $

\end{enumerate}
}
\end{definition}

Since ${\mathcal D}({\mathcal A}_I)$ contains the subdirectly
irreducible algebras of ${\mathcal A}_I$, we have that: $${{\mathcal
D}({\mathcal A}_I)} \models t=s \hspace{0.3cm} \mathrm{iff}
\hspace{0.3cm} {{\mathcal A}_I}\models t=s $$ where $t,s$ are terms
in the language of ${\mathcal A}_I$.

Thus, when we say that a subclass ${\mathcal A}$ of ${\mathcal E}_B$
is equationally characterizable by a subvariety ${\mathcal A}_I$ of
${\mathcal IE}_B$ this means that the objects of ${\mathcal A}$ are
identifiable with the  directly indecomposable algebras
of ${\mathcal A}_I$ according to the items {\bf I} and {\bf E} in Definition \ref{DEF1}.\\

Taking into account the concept of $IE_B$-lattice we introduce a way
to study the notion of Boolean$^*$ pre-state given in Definition
\ref{BOOLEANSTAR} via a  unary operation added to the Baer
$^*$-semigroups structure.

In fact, let $S$ be a Baer $^*$-semigroup. A unary operation $s$ on
$S$ that allows us to capture the notion of  Boolean$^*$ pre-state
would have to satisfy the following basic conditions:

\begin{itemize}

\item[a.]
$s(x') = s(x)'$.

\item[b.]
The restriction $s/_{P_c(S)}$ defines a unary operation in $P_c(S)$
such that $\langle P_c(S), \lor,\land, ', s/_{P_c(S)}, 0,1 \rangle$
is an $IE_B-{\mathrm lattice}$.

\item[c.]
$s$ should satisfy a version of Theorem \ref{BSIGMA3} i.e., $s$ should
be always obtainable as the unique extension of $s/_{P_c(S)}$.

\end{itemize}

These ideas motivate the following general definition:

\begin{definition}\label{E}
{\rm An $IE^*_B$-semigroup is an algebra $ \langle S,
\cdot, ^*, ', s, 0 \rangle$ of type $ \langle 2, 1,1, 1, 0 \rangle$
such that $ \langle S, \cdot, ^*, ', 0  \rangle$ is a
Baer $^*$-semigroup and $s$ satisfies the following equations for
each $x,y \in S$:

\begin{enumerate}

\item[\rm{bs1.}]
$s(1) = 1$,

\item[\rm{bs2.}]
$s(x') = s(x)'$,

\item[\rm{bs3.}]
$s(x)'' = s(x)$,

\item[\rm{bs4.}]
$s(x' \lor s(y')) = s(x') \lor s(y')$,

\item[\rm{bs5.}]
$y' = (y' \land s(x)) \lor (y' \land s(x)') $,

\item[\rm{bs6.}]
$s(x' \land y') \leq s(x')\land s(y') $.

\end{enumerate}
}

\end{definition}

\noindent Thus, the class of $IE^*_B$-semigroups is a variety that
we call ${\mathcal IE}^*_B$.

\begin{prop}\label{BS2}
Let $S$ be an $IE^*_B$-semigroup. Then

\begin{enumerate}
\item
$s(x) \in Z(P_c(S))$.

\item
$\langle P_c(S), \lor, \land, ', s/_{P_c(S)}, 0,1  \rangle$ is an
$IE_B$-lattice and $\langle s(S), \lor, \land, ', 0,1 \rangle$ is a
Boolean subalgebra of $Z(P_c(S))$.

\item
$s(x'') = s(x)$.

\item
If $x,y \in P(S)$ and $x\leq y$ then, $s(y')\leq s(x')$ and  $s(x)\leq s(y)$.

\item
$s(x \cdot y) = s(x'' \cdot y) \leq s(y)$.

\item
$s(x^* \cdot x ) = s(x)$.

\end{enumerate}
\end{prop}

\begin{proof}
1 and 2) By bs3, for each $x\in S$, $s(x) \in P_c(S)$. Then, by
Proposition \ref{eqcentro}-2 and bs5, $s(x) \in Z(P_c(S))$. Since
the image of $'$ is $P_c(S)$, from the rest of the axioms, $\langle
P_c(S), \lor, \land, ', s/_{P_c(S)}, 0,1  \rangle$ is an
$IE_B$-lattice and  $\langle s(S), \lor, \land, ', 0,1 \rangle$ is a
Boolean subalgebra of $Z(P_c(S))$. 3,4,5,6) Follow from similar
arguments used in the proof of Proposition \ref{BS1}.
\rule{5pt}{5pt}
\end{proof}

\begin{theo}\label{BSS3}
Let $S$ be a Baer $^*$-semigroup and $\langle P_c(S), \lor, \land,
', s, 0,1 \rangle$ be an $IE_B$-lattice. Then the operation $s_S:S
\rightarrow S $ such that: $$s_S(x) = s(x'')$$ defines the unique
$IE^*_B$-semigroup structure on $S$ such that $s_S /_{P_c(S)} = s $.
\end{theo}

\begin{proof}
If $x\in S$ then $x'' \in P_c(S)$ and $s(x'')$ is defined. Then
$s_S$ is well defined as a function.  Since $x\in P_c(S)$ iff $x=x''$, $s_S(x) = s(x'') = s(x)$ for each $x\in P_c(S)$.
Thus $s_S /_{P_c(S)} = s$. Now we prove the validity of the axioms bs1,...,bs6.

bs1) Is immediate. bs2) $s_S(x') = s(x''') = s(x'')' = s_S(x)'$.
bs3) $s_S(x)'' = s(x'')'' = s(x'')$ since $s(x) \in P_c(S)$ and $'$
is an orthocomplementation on $ P_c(S)$. Hence $s_S(x)'' = s_S(x)$.
bs4, bs5, bs6) Follow from the fact that $s_S /_{P_c(S)} = s$ and
$\langle P_c(S), \lor, \land, ', s, 0,1  \rangle$ is an
$IE_B$-lattice. Hence $s_S$ defines an $IE^*_B$-semigroup structure
on $S$ such that $s_S /_{P_c(S)} = s$.

Suppose that $ \langle S, \cdot, ^*, ', s_1, 0 \rangle$ is an
$IE^*_B$-semigroup such that $s_1 /_{P_c(S)} = s $. Let $x\in S$.
Since $x''\in P_c(S)$, by Proposition \ref{BS2}-3, $s_1(x) =
s_1(x'') = s(x'') = s_S(x)$. Hence $s_1 = s_S$ and $s_S$ defines the
unique $IE^*_B$-semigroup structure on $S$ such that $s_S /_{P_c(S)}
= s $. \rule{5pt}{5pt}
\end{proof}

By Proposition \ref{BS2} and Theorem \ref{BSS3} we can see that the
definition of  $IE^*_B$-semigroup pre-state satisfies the
conditions required by items a, b, c.

\begin{coro}\label{BS4}
Let $\langle L, \lor, \land, \neg, s, 0,1 \rangle$ be an
$IE_B$-lattice and $S(L)$ be the Baer $^*$-semigroup of residuated
functions of $L$. If for each Sasaki projection $\phi_a$ we define
$\bar{s}(\phi_a) = \phi_{s(a)}$, then:

\begin{enumerate}
\item

$\langle P_c(S(L)), \lor, \land,' , \bar{s}, 0,1 \rangle$ is an
$IE_B$-lattice and $f:L \rightarrow P_c(S(L))$ such that $f(a) =
\phi_a$ is an $IE_B$-isomorphism.

\item
The operation ${\bar{s}}_S(\varphi) = \phi_{s(\varphi(1))}$ defines
the unique $IE^*_B$-semigroup structure on $S(L)$ such that $L$ is
$IE_B$-isomorphic to $P_c(S(L))$.

\end{enumerate}
\end{coro}

\begin{proof}
1) By Theorem \ref{PRO2},  there exists an $OML$-isomorphism $f:L
\rightarrow P_c(S(L))$. It is not very hard to see that the
composition $\bar{s} = f s f^{-1}$ satisfies bs1,...,bs6 and
$f(s(x)) = (f s  f^{-1}) f(x) = \bar{s}f(x)$, i.e. $f$ preserves
$\bar{s}$. Then $L$ is $IE_B$-isomorphic to $P_c(S(L))$.

2) Let $\varphi \in S(L)$. Then  ${\bar{s}}_S(\varphi) =
\phi_{s(\varphi(1))} = \bar{s}(\phi_{\varphi(1)}) =
\bar{s}(\phi_{\neg \neg \varphi(1)}) = s(\varphi'')$. Therefore
${\bar{s}}_S$ is the extension of $\bar{s}$ given in Theorem
\ref{BSS3}. Hence the operation ${\bar{s}}_S$ defines the unique
$IE^*_B$-semigroup structure on $S(L)$ such that $L$ is
$IE_B$-isomorphic to $P_c(S(L))$. \rule{5pt}{5pt}
\end{proof}

\begin{coro}\label{BS5}
Let $ \langle S, \cdot, ^*, ', s, 0 \rangle$ be an
$IE^*_B$-semigroup. Suppose that $S_1$ is a sub Baer $^*$-semigroup
of $S$ and $P_c(S_1)$ is a sub $IE_B$-lattice of $P_c(S)$. Then the
restriction $s/_{S_1}$ defines the unique $IE^*_B$-semigroup
structure on $S_1$. In this way, $S_1$ is also a sub
$IE^*_B$-semigroup of $S$.
\end{coro}

\begin{proof}
Let $S_1$ be a sub Baer $^*$-semigroup of $S$ such that $P_c(S_1)$
is a sub $IE_B$-lattice of $P_c(S)$. If for each $x\in S_1$ we
define $s_{S_1}(x) = s/_{P_c(S_1)}(x'') = s(x'')$ then, $s_{S_1} =
s/_{S_1}$ and, by Theorem \ref{BSS3}, it is defines the unique
$IE^*_B$-semigroup structure on $S_1$ that coincides with
$s/_{P_c(S_1)}$ in $P_c(S_1)$. In this way $S_1$ also results  a sub
$IE^*_B$-semigroup of $S$. \rule{5pt}{5pt}
\end{proof}

\begin{prop}\label{PRO11}
Suppose that $(S_i)_{i\in I}$ is a family of $IE^*_B$-semigroups.
Then, $\prod_{i\in I}P_c(S_i)$ is $IE_B$-lattice isomorphic to
$P_c(\prod_{i\in I}S_i)$.

\end{prop}

\begin{proof}
Since the operations in $\prod_{i\in I}S_i$ are defined pointwise,
for each $(x_i)_{i\in I} \in \prod_{i\in I}S_i$, $(x_i)'_{i\in I} =
(x'_i)_{i\in I}$. Then it is straightforward to prove that
$f((x_i)'_{i\in I}) = (x'_i)_{i\in I}$ defines an $OML$-isomorphism
$f: P_c(\prod_{i\in I}S_i) \rightarrow \prod_{i\in I}P_c(S_i)$. We
have to prove that this function preserves $s$. In fact
$f(s((x_i)'_{i\in I})) = f( (s(x_i))'_{i\in I}) = (s(x_i)')_{i\in I}
= (s(x'_i))_{i\in I} = s((x'_i)_{i\in I}) = s(f((x_i)'_{i\in I}))$.
Hence $f$ is an $IE_B$-lattice isomorphism. \rule{5pt}{5pt}
\end{proof}

In what follow we study the relation between Boolean$^*$ pre-states
and $IE^*_B$-semigroups.

\begin{prop}\label{FUNC32}
Let $ \langle S, \cdot, ^*, ', s, 0 \rangle$ be an $IE^*_B$-semigroup.
Then there exists a Boolean$^*$ pre-state  $\sigma:S \rightarrow \{0,1\}$ such that $s/_{P_c(S)}$ is
coherent with $\sigma/_{P_c(S)}$.
\end{prop}

\begin{proof}
By Proposition \ref{FUNC00} there exists a Boolean pre-state
$\sigma_0: P_c(S) \rightarrow \{0,1\} $ such that  $s/_{P_c(S)}$ is
coherent with $\sigma_0$. By Theorem  \ref{BSIGMA3}, there exists a
unique Boolean$^*$ pre-state  $\sigma:S \rightarrow \{0,1\}$ such
that  $\sigma_0 = \sigma/_{P_c(S)}$. Hence $s/_{P_c(S)}$ is coherent
with $\sigma/_{P_c(S)} = \sigma_0$. \rule{5pt}{5pt}
\end{proof}

The following result gives a kind of converse of the last
proposition:

\begin{prop}\label{FUNC33}
Let $S$ be a Baer $^*$-semigroup and  $\sigma:S \rightarrow \{0,1\}$ be
a Boolean$^*$ pre-state. If we define $$ s_{\sigma}(x) = \cases
{1^{P_c(S)}, & if $\sigma(x)=1$ \cr 0^{P_c(S)} , & if $\sigma(x)=0$ \cr} $$then
$ \langle S, \cdot, ^*, ', s_{\sigma}, 0 \rangle$ is an $IE^*_B$-semigroup
and $s_\sigma/_{P_c(S)}$ is coherent with $\sigma/_{P_c(S)}$

\end{prop}

\begin{proof}
By Proposition \ref{FUNC0}, $\langle P_c(S), \land, \lor,',
s_\sigma/_{P_c(S)}, 0,1  \rangle$ is an $IE_B$-lattice and
$s_\sigma/_{P_c(S)}$ is coherent with $\sigma/_{P_c(S)}$. Since
$\sigma(x) = \sigma(x'')$ then $s_\sigma(x) = s_\sigma(x'') =
s_\sigma/_{P_c(S)}(x'')$. Hence, by Theorem \ref{BSS3}, $s_\sigma$
defines the unique $IE^*_B$-semigroup structure on $S_1$ that
extends $s_\sigma/_{P_c(S)}$. \rule{5pt}{5pt}
\end{proof}

\section{Varieties of $IE_B$-lattices determining varieties of $IE^*_B$-semigroups} \label{VARIETIES}

When a family of two-valued states over an orthomodular lattice is
equationally characterizable  by a variety of $IE_B$-lattices in the
sense of Definition \ref{DEF1}, the problem about the existence of a
variety of $IE^*_B$-semigroups that somehow may be able to
equationally characterize the mentioned family of two-valued states
may be posed. The following definition provides a ``natural
candidate'' for such a class of $IE^*_B$-semigroups.

\begin{definition}\label{AIASTE}
{\rm Let ${\mathcal A}_I$ be a subvariety of ${\mathcal IE}_I$. Then
we define the subclass ${\mathcal A}^*_I$ of  ${\mathcal IE}^*_I$ as
follows:
$${\mathcal A}^*_I =  \{S \in {\mathcal IE}^*_B: P_c(S)\in {\mathcal A}_I  \}$$

}
\end{definition}

Before proceeding,  we have to make sure that ${\mathcal A}^*_I$ is
a non-empty subclass of ${\mathcal IE}^*_I$.

\begin{prop}\label{NONEMPTY}
If ${\mathcal A}_I$ is a non-empty subvariety of ${\mathcal IE}_I$
then ${\mathcal A}^*_I$ is a non-empty subclass of ${\mathcal
IE}^*_B$.

\end{prop}

\begin{proof}
Suppose that $ \langle L, \land, \lor, \neg, s, 0,1 \rangle$ belong
to ${\mathcal A}_I$. By Theorem \ref{PRO2} we can consider the Baer
$^*$-semigroup $S(L)$ of residuated functions in $L$ in which $L$ is
$OML$-isomorphic to $P_c(S(L))$. Identifying $L$ with $P_c(S(L))$,
by Theorem \ref{BSS3}, there exists an operation $s_{S(L)}$ on
$S(L)$ that defines the unique $IE^*_B$-semigroup structure on
$S(L)$ such that $s_{S(L)}/_{P_c(S(L))} = s$. Hence $S(L) \in
{\mathcal A}^*_I$ and ${\mathcal A}^*_I$ is a non empty subclass of
${\mathcal IE}^*_B$. \rule{5pt}{5pt}
\end{proof}

In what follows we shall demonstrate not only that  ${\mathcal
A}^*_I$ is a variety but also we shall give a decidable method to
find  an equational system that defines ${\mathcal A}^*_I$ from an
equational system that defines ${\mathcal A}_I$. In order to study
this we first introduce the following concept:

\begin{definition}\label{TRANS}
{\rm We define the {\em $*$-translation} $\tau:$ Term$_{{\mathcal
IE}_B} \rightarrow$ Term$_{{\mathcal IE}^*_B}$ as follows:

\begin{itemize}

\item[]
$\tau(0) = 0$ and $\tau(1) = 1$

\item[]
$\tau(x) = x'$ for each variable $x$,

\item[]
$\tau(\neg t) = \tau(t)'$,

\item[]
$\tau (t \land s) = (\tau(t)'\cdot \tau(s))' \cdot \tau(s)$,

\item[]
$\tau(t \lor s) = \tau(\neg(\neg t \land \neg s))$,

\item[]
$\tau(s(t)) = s(\tau(t))$.

\end{itemize}
}
\end{definition}

\begin{prop}\label{EXTENS}
Let $S$ be a $IE^*_B$-semigroup, $v:$ {\rm Term}$_{{\mathcal
IE}^*_B}\rightarrow S $ be a valuation and $\tau$ be the
$*$-translation. Then:

\begin{enumerate}
\item
For each $t\in$ {\rm Term}$_{{\mathcal IE}_B}$,  $v(\tau(t)) \in
P_c(S)$,

\item
There exists a valuation $v_c:$ {\rm Term}$_{{\mathcal IE}_B}
\rightarrow P_c(S)$ such that for each $t\in$ {\rm Term}$_{{\mathcal
IE}_B}$, $v_c(t) = v(\tau(t))$.

\end{enumerate}
\end{prop}

\begin{proof}
1) Let $t\in$ Term$_{{\mathcal IE}_B}$.  If $t$ is the form $\neg r$
then,  $v(\tau(r)) = v(\tau(\neg r)) = v(\tau(r)') = v(\tau(r))' \in
P_c(S)$. If $t$ is the form $s(r)$ then, $v(\tau(s(r))) =
v(s(\tau(r))) = s(v(\tau(r))) \in Z(P_c(S)) \subseteq P_c(S)$.  For
the other case we use induction on the complexity of terms in
Term$_{{\mathcal IE}_B}$. If Comp($t$)$ = 0$ then $t$ is $0$, $1$,
or a variable $x$. In these cases $v(\tau(1)) = v(1) = 1^S$,
$v(\tau(0)) = v(0) = 0^S$ and $v(\tau(x)) = v(x') = v(x)'$. Thus
$v(\tau(t)) \in P_c(S)$. Assume that $v(\tau(t)) \in P_c(S)$
whenever Comp($t$)$ < n$. Suppose that Comp($t$)$ = n$. We have to
consider the case in which $t$ is the form $p \land r$. Then
$v(\tau(t)) = v(\tau(p\land r)) = v((\tau(p)'\cdot \tau(r))' \cdot
\tau(r)) = (v(\tau(p))'\cdot v(\tau(r)))' \cdot v(\tau(r)) =
v(\tau(r))\land v(\tau(p))$ because $v(\tau(r))\in  P_c(S)$ and
$v(\tau(p)) \in  P_c(S)$. Thus $v(\tau(t)) \in P_c(S)$. This proves
that for each $t\in$ Term$_{{\mathcal IE}_B}$
$v(\tau(t)) \in P_c(S)$. \\

2) Consider the valuation $v_c:$ Term$_{{\mathcal IE}_B} \rightarrow
P_c(S)$ such that for each variable $x$, $v_c(x) = v(x')$. Now we
proceed by induction on the complexity of terms in Term$_{{\mathcal
IE}_B}$. If Comp($t$)$ = 0$ then $t$ is $0$, $1$, or a variable $x$.
Then, $v_c(1) = 1^S = v(1) = v(\tau(1))$, $v_c(0) = 0^S = v(0) =
v(\tau(0))$ and  $v_c(x) = v(x') = v(\tau(x))$.  Assume that $v_c(t)
= v(\tau(t))$ whenever Comp($t$)$ < n$. Suppose that Comp($t$)$ =
n$. We have to consider three possible cases:

$t$ is the form $\neg r$. Then $v_c(t) = v_c(\neg r) = v_c(r)' = v(\tau(r))' = v(\tau(r)') = v(\tau(\neg r)) = v(\tau(t))$.

$t$ is the form $p \land r$.  $v_c(t) = v_c(p\land r) = v_c(p) \land
v_c(r) = v(\tau(p)) \land v(\tau(r)) = (v(\tau(p))'\cdot
v(\tau(r)))' \cdot v(\tau(r)) = v((\tau(p)'\cdot \tau(r))' \cdot
\tau(r)) = v(\tau(p\land r)) = v(\tau(t))$.

$t$ is the form $s(r)$. Then $v_c(t) = v_c(s(r)) = s(v_c(r)) = s(v(\tau(r))) = v(s(\tau(r))) = v(\tau(s(r))) = v(\tau(t))$.

This proves that  for each $t\in$ Term$_{{\mathcal IE}_B}$, $v_c(t)
= v(\tau(t))$. \rule{5pt}{5pt}
\end{proof}\\

\begin{prop}\label{EXTENS2}
Let $S$ be an $IE^*_B$-semigroup and $v:$  {\rm Term}$_{{\mathcal
IE}_B} \rightarrow P_c(S) $ be a valuation. Then there exists a
valuation $v^*:$  {\rm Term}$_{{\mathcal IE}^*_B} \rightarrow S $
such that $t\in$ {\rm Term}$_{{\mathcal IE}_B}$, $v^*(\tau(t)) =
v(t)$.

\end{prop}

\begin{proof}
Consider the valuation $v^*:$ Term$_{{\mathcal IE}^*_B} \rightarrow
S $ such that for each variable $v^*(x) = v(\neg x)$. Let $t\in$
Term$_{{\mathcal IE}_B}$. We use induction on the complexity of
terms in Term$_{{\mathcal IE}_B}$. If Comp($t$)$ = 0$ then $t$ is
$0$, $1$, or a variable $x$. Then, $v^*(\tau(1)) = v^*(1) = 1^S =
v(1)$, $v^*(\tau(0)) = v^*(0) = 0^S = v(0)$ and $v^*(\tau(x)) =
v^*(x') = v^*(x)' = v(\neg x)' = v(x)'' = v(x)$ since $v(x) \in
P_c(S)$. Assume that $v^*(\tau(t)) =  v(t) $ whenever Comp($t$)$ <
n$. Suppose that Comp($t$)$ = n$. We have to consider three possible
cases:

$t$ is the form $\neg r$. Then $v^*(\tau(t)) = v^*(\tau(\neg r)) = v^*(\tau(r)') = v^*(\tau(r))' = v(r)' = v(\neg r) = v(t)$.

$t$ is the form $p \land r$. Then $v^*(\tau(t)) = v^*(\tau(p \land
r)) = v^*( (\tau(p)'\cdot \tau(r))' \cdot \tau(r)) =
(v^*(\tau(p))'\cdot v^*(\tau(r)))' \cdot v^*(\tau(r)) = (v(p)' \cdot
v(r))' \cdot v(r) = v(p)\land v(r) = v(p\land r) = v(t)$.

$t$ is the form $s(r)$. Then $v^*(\tau(t)) = v^*(\tau(s(r))) = v^*(s(\tau(r))) = s(v^*(\tau(r))) = s(v(r)) = v(s(r)) = v(t)$.

This proves that for each $t\in$ Term$_{{\mathcal IE}_B}$,
$v^*(\tau(t)) = v(t)$. \rule{5pt}{5pt}
\end{proof}\\

\begin{theo} \label{1EQ}
Let ${\mathcal A}_I$ be a subvariety of ${\mathcal IE}_I$ and assume
that $\{t_i = s_i \}_{i\in I}$ is a set of equations in the language
of ${\mathcal IE}_I$ that defines ${\mathcal A}_I$. Then $${\mathcal
A}^*_I = \{S \in {\mathcal IE}^*_B: \forall i\in I, S \models
\tau(t_i) = \tau(s_i) \}$$

\end{theo}

\begin{proof}
On the one hand, assume that $S \in {\mathcal A}^*_I$ i.e.,
$P_c(S)\models t_i = s_i$ for each $i\in I$. Suppose that there
exists $i_0 \in I$ such that $S \not \models \tau(t_{i_0}) =
\tau(s_{i_0})$. Then there exists a valuation $v:$ Term$_{{\mathcal
IB}_B} \rightarrow S $ such that $v(\tau(t_{i_0})) \not =
v(\tau(s_{i_0}))$. By Proposition \ref{EXTENS} there exists a
valuation $v_c:$ Term$_{{\mathcal IE}_B} \rightarrow P_c(S)$ such
that for each $t\in$ Term$_{{\mathcal IE}_B}$, $v_c(t) =
v(\tau(t))$. Then $v_c(t_{i_0}) = v(\tau(t_{i_0})) \not =
v(\tau(s_{i_0})) = v_c(s_{i_0})$ and  $P_c(S)\not \models t_{i_0} =
s_{i_0}$ which is a contradiction. Hence $S \models \tau(t_i) =
\tau(s_i)$ for each $i\in I$.

On the other hand, assume that $S \in {\mathcal IE}^*_B$ and $S
\models \tau(t_i) = \tau(s_i)$ for each $i\in I$. Suppose that there
exists $i_0 \in I$ such that $P_c(S)\not \models  t_{i_0} =
s_{i_0}$. Then there exists a valuation $v:$  Term$_{{\mathcal
IE}_B} \rightarrow P_c(S)$ such that $v(t_{i_0}) \not = v(s_{i_0})$.
By Proposition \ref{EXTENS2}, there exists a valuation $v^*:$
Term$_{{\mathcal IE}^*_B} \rightarrow S $ such that for each $t\in$
Term$_{{\mathcal IE}_B}$, $v^*(\tau(t)) = v(t)$. Then $
v^*(\tau(t_{i_0})) = v(t_{i_0}) \not = v(s_{i_0}) =
v^*(\tau(s_{i_0})) $ and $S \not \models \tau(t_{i_0}) =
\tau(s_{i_0})$ which is a contradiction. Hence, for each $i\in I$,
we have $P_c(S)\models t_i = s_i$. \rule{5pt}{5pt}
\end{proof}

\section{Baer $^*$-semigroups and the full class of two-valued states}  \label{FULLCLASS}

The category of orthomodular lattices admitting a two-valued state,
noted by ${\mathcal TE}_B$, is the full sub-category of ${\mathcal
E}_B$ whose objects are $E_B$-lattices $(L,\sigma)$ satisfying the
following condition: $$ x\bot y \hspace{0.2cm} \Longrightarrow
\hspace{0.2cm} \sigma(x\lor y) = \sigma(x)+\sigma(y). $$ In
\cite[Theorem 6.3]{DFD} it is proved that the variety
$${\mathcal ITE}_B = {\mathcal IE}_B + \{s(x \lor (y \land \neg x) )= s(x) \lor s(y \land \neg x) \} $$
equationally characterizes ${\mathcal TE}_B$ in the sense of
Definition \ref{DEF1}. Thus, the objects of ${\mathcal TE}_B$ are
identifiable to the directly indecomposable algebras of ${\mathcal
ITE}_B$. By Theorem \ref{1EQ} we can give an equational theory in
the frame of Baer $^*$-semigroups that captures the concept of
two-valued state. In fact this is done through the variety
$${\mathcal ITE}^*_B = {\mathcal IE}^*_B + \{s(x' \lor (y' \land x'') )=
s(x') \lor s(y' \land x'') \} $$

\section{Baer $^*$-semigroups and Jauch-Piron two-valued states} \label{JAUCHPIRON}

The category of orthomodular lattices admiting a Jauch-Piron
two-valued state \cite{RU}, noted by ${\mathcal JPE}_B$, is the full
sub-category of ${\mathcal TE}_B$ whose objects are $E_B$-lattices
$(L,\sigma)$ in ${\mathcal TE}_B$ also satisfying the following
condition:
$$\sigma(x) = \sigma(y) = 1 \hspace{0.3cm} \Longrightarrow
\hspace{0.3cm} \sigma(x\land y) = 1  $$ In \cite[Theorem 7.3]{DFD}
it is proved that the variety $${\mathcal IJPE}_B = {\mathcal ITE}_B
+ \{s(x) \land s(\neg x \lor y) = s(x\land y) \} $$ equationally
characterizes ${\mathcal JPE}_B$ in the sense of Definition
\ref{DEF1}. Thus the objects of ${\mathcal JPE}_B$ are identifiable
to the  directly indecomposable algebras of ${\mathcal IJPE}_B$.  By
Theorem \ref{1EQ} we can give an equational theory in the frame of
Baer $^*$-semigroups that captures the concept of two-valued state.
In fact this is done through the variety $${\mathcal IJPE}^*_B =
{\mathcal ITE}^*_B + \{ s(x') \land s(x'' \lor y') = s(x'\land y')\}
$$

\section{The problem of equational completeness in ${\mathcal A}^*_I$} \label{APROBLEM}

Let ${\mathcal A}$ be a family of $E_B$-lattices. Suppose that the
subvariety ${\mathcal A}_I$ of ${\mathcal IE}_I$ equationally
characterizes ${\mathcal A}$ in the sense of Definition \ref{DEF1}.
Then, through a functor  ${\mathcal I}$, ${\mathcal A}$ is
identifiable to the directly indecomposable algebras of the variety
${\mathcal A}_I$. In this way, we can state that  ${\mathcal A}$
determines the equational theory of ${\mathcal A}_I$. With the
natural extension of Boolean pre-states to Baer $^*$-semigroups,
encoded in ${\mathcal A}^*_I$, this kind of characterization may be
lost. More precisely, the class ${\mathcal A}$ may ``not rule'' the
equational theory of ${\mathcal A}^*_I$ in the way ${\mathcal A}$
does with ${\mathcal A}_I$. The following example shows such a
situation:

\begin{example}\label{problem1}
{\rm Let $\widetilde{{\mathcal B}}$ be the subclass of ${\mathcal
E}_B$ formed by the pairs $(B,\sigma)$ such that $B$ is a Boolean
algebra and $\sigma$ is a Boolean pre-state.  $\widetilde{{\mathcal
B}}$ is a non empty class since Boolean homomorphisms of the form
$B\rightarrow {\mathbf 2} $ always exist  for each Boolean algebra
$B$ and they are examples of Boolean pre-states. It is clear that
the class
$${\widetilde{{\mathcal B}}}_I = {\mathcal IE}_B + \{x \land (y \lor z) =
(x\land y) \lor (x \land z ) \} $$ equationally characterizes  the
class $\widetilde{{\mathcal B}}$ in the sense of Definition
\ref{DEF1}. Note that ${\widetilde{{\mathcal B}}}_I$ may be seen as
a sub-variety ${\widetilde{{\mathcal B}}}^*_I$  since, each algebra
$B$ in ${\widetilde{{\mathcal B}}}_I$ in the signature $\langle
\land,
*, \neg, s, 0 \rangle$ where $*$ is the identity, is a
$IE^*_B$-semigroup. Then the equational theory of
${\widetilde{{\mathcal B}}}_I$, as variety of $IE^*_B$-semigroups,
is determined by the algebras of $\widetilde{{\mathcal B}}$. Note
that algebras of ${\widetilde{{\mathcal B}}}_I$ are commutative Baer
$^*$-semigroups and then we have
$${\widetilde{{\mathcal B}}}_I \models x\cdot y = y \cdot x$$ What we want to point out is the following: ${\widetilde{{\mathcal
B}}}_I$ captures (although in some sense a trivial one) the concept
of Boolean pre-states over Boolean algebras in a variety. Moreover
$\widetilde{{\mathcal B}}$ also determines the equational theory of
${\widetilde{{\mathcal B}}}_I$ when ${\widetilde{{\mathcal B}}}_I$
is seen as  a variety of $IE^*_B$-semigroup.

Let us now compare  the last result with Definition \ref{TRANS} and
Theorem \ref{1EQ}. The variety ${\widetilde{{\mathcal B}}}^*_I$
given by
\begin{eqnarray*}
{\widetilde{{\mathcal B}}}^*_I & = & {\mathcal IE}^*_B + \{ \tau(x \land (y \lor z)) = \tau ((x\land y) \lor (x \land z )) \} \\
& = & {\mathcal IE}^*_B + \{x' \land (y' \lor z') = (x'\land y')
\lor (x' \land z' ) \}
\end{eqnarray*}
is the biggest  subvariety of ${\mathcal IE}^*_B$ whose algebras
have a lattice of closed projections with Boolean structure and then
${\widetilde{{\mathcal B}}}_I  \subseteq {\widetilde{{\mathcal
B}}}^*_I$. We shall prove that ${\widetilde{{\mathcal B}}}_I  \not =
{\widetilde{{\mathcal B}}}^*_I$, i.e. the inclusion,  is proper. In
fact:

Let $B_4$ be the Boolean algebra of four elements $\{0, a, \neg a, 1
\}$ endowed with the operation $s(x) = x$ i.e., the identity on
$B_4$. In this case $B_4 \in {\widetilde{{\mathcal B}}}_I$.
According to Theorem \ref{PRO2}, we consider the  Baer
$^*$-semigroup $S(B_4)$ of residuated functions of $B_4$. Since we
can identify $B_4$ with $P_c(S(B_4))$, by Proposition \ref{BSS3} we
can extend $s$ to $S(B_4)$. Therefore $S(B_4)$ may be seen as an
algebra of ${\widetilde{{\mathcal B}}}^*_I$. Consider the function
$\phi: B_4 \rightarrow B_4$ such that $0\phi =\phi(0) = 0$, $1\phi
=\phi(1) = 1$, $a\phi = \phi(a) = \neg a$ and $(\neg a)\phi =
\phi(\neg a) = a$. Note that $\phi$ is an order preserving function
and the composition $\phi  \phi = 1_{B_4}$. Hence $\phi$ is the
residual function of itself and then $\phi \in S(B_4)$. Let $\phi_a$
be the Sasaki projection associated to $a$. Then $x\phi_a = (x\lor
\neg a) \land a = x \land a$. Note that $a(\phi \phi_a) =
(a\phi)\land a = \neg a \land a = 0$ and $a(\phi_a \phi) =
(a\phi_a)\phi = a\phi = \neg a$.

This proves that $\phi \phi_a \not = \phi_a \phi$ and then, $S(B_4)
\not \in {\widetilde{{\mathcal B}}}_I$ because
${\widetilde{{\mathcal B}}}_I$ is a variety of commutative Baer
$^*$-semigroups. Thus ${\widetilde{{\mathcal B}}}_I \not =
{\widetilde{{\mathcal B}}}^*_I$ and the inclusion is proper.

Hence the directly indecomposable algebras of ${\widetilde{{\mathcal
B}}}_I$ considered as $IE^*_B$-semigroups do not determine the
equational theory of ${\widetilde{{\mathcal B}}}^*_I$. Consequently
$\widetilde{{\mathcal B}}$ does not significatively add to the
equational theory of ${\widetilde{{\mathcal B}}}^*_I$.

}
\end{example}

Taking into account Example \ref{problem1}, the following problem
may be posed:

\begin{itemize}
\item[]
{\it Let  ${\mathcal A}$ be a class of $E_B$-lattices and suppose
that the subvariety ${\mathcal A}_I$ of ${\mathcal IE}_I$
equationally characterizes  ${\mathcal A}$ in the sense of
Definition \ref{DEF1}. Give a subvariety ${\mathcal G}^*$ of
${\mathcal A}^*_I$ in which we can determine the equational theory
of ${\mathcal G}^*$ from the class ${\mathcal A}$.}
\end{itemize}

We conclude this section defining the meaning of the statement that,
the class ${\mathcal A}$ of $E_B$-lattices determines the equational
theory of a subvariety of ${\mathcal A}^*_I$.

\begin{definition} \label{problem3}
{\rm Let  ${\mathcal A}$ be a class of $E_B$-lattices. Suppose that
the variety  ${\mathcal A}_I$ of  $IE_B$-lattices equationally
characterizes the class ${\mathcal A}$ and ${\mathcal I}: {\mathcal
A} \rightarrow {\mathcal D}({\mathcal A}_I)$ is the functor that
provides the categorical equivalence between ${\mathcal A}$ and the
category ${\mathcal D}({\mathcal A}_I)$ of directly indecomposable
algebras of ${\mathcal A}_I$. We say that ${\mathcal A}$ determines
the equational theory of a subvariety ${\mathcal G}^*$ of ${\mathcal
A}^*_I$ iff there exists a class operator $${\mathbb G} : {\mathcal
A}_I \rightarrow {\mathcal A}^*_I$$ such that:

\begin{enumerate}
\item
For each $L \in {\mathcal A}$, ${\mathbb G}{\mathcal I}(L)$ is a
directly indecomposable algebra in ${\mathcal A}^*_I$.

\item
${\mathcal G}^* = {\mathcal V}(\{{\mathbb G}{\mathcal I}(L): L \in
{\mathcal A} \})$.

\end{enumerate}
}
\end{definition}

In the next section, we will study a class operator, denoted by
${\mathbb S}_0$, that will allow us to define a subvariety of
${\mathcal A}^*_I$ whose equational theory is determinated by
${\mathcal A}$ in the sense of Definition \ref{problem3}.

\section{The class operator ${\mathbb S}_0$} \label{OPERATORE}

Let $\langle L, \land, \lor, \neg, s, 0,1 \rangle$ be an
$IE_B$-lattice. By Corollary \ref{BS4}, we consider the
$IE^*_B$-semigroup $S(L)$ of residuated functions of $L$. By abuse
of notation, we also denote by $s$ the operation $s_{S(L)}$ on
$S(L)$ where $s_{S(L)}(x) = s(x'')$. Let $S_0(L)$ be the sub Baer
$^*$-semigroup of $S(L)$ generated by the Sasaki projections on $L$.
In the literature, $S_0(L)$ is refereed as {\em the small  Baer
$^*$-semigroup of products of Sasaki projections on $L$} \cite{AD,
FOU}. By Corollary \ref{BS5}, $S_0(L)$ with the restriction
$s/_{S_0(L)}$ is a sub  $IE^*_B$-semigroup of $S(L)$. This
$IE^*_B$-semigroup will be denoted by ${\mathbb S}_0(L)$.

Since $P_c({\mathbb S}_0(L))$ is $IE_B$-isomorphic to $L$, if
${\mathcal A}_I$ is a subvariety of ${\mathcal IE}_I$ and $L \in
{\mathcal A}_I$ then ${\mathbb S}_0(L) \in {\mathcal A}^*_I$. These
results allow us to define the following class operator: $${\mathbb
S}_0: {\mathcal A}_I \rightarrow {\mathcal A}^*_I \hspace{0.4cm}s.t.
\hspace{0.2cm} L\mapsto {\mathbb S}_0(L) $$

\begin{prop} \label{HOMBAER}
Let $L_1, L_2$ be two $IE_B$-lattices and $f:L_1 \rightarrow L_2$ be
an $IE_B$-homomorphism.

\begin{enumerate}

\item
If $\phi_{a_1} \ldots \phi_{a_n}$ are Sasaki projections in $L_1$
then for each $x \in L_1$ we have $f(x\phi_{a_1} \ldots \phi_{a_n})
= f(x) \phi_{f(a_1)} \ldots \phi_{f(a_n)} $.

\item
If $f$ is a surjective function then there exists an unique
$IE^*_B$-homomorphism $g: {\mathbb S}_0(L_1) \rightarrow {\mathbb
S}_0(L_2)$ such that, identifying $L_1$ with $P_c({\mathbb
S}_0(L_1))$, $g/_{L_1} = f$  Moreover, $g$ is a surjective function.

\item
If $f$ is bijective then ${\mathbb S}_0(L_1)$ and ${\mathbb
S}_0(L_2)$ are $IE^*_B$-isomorphic.

\end{enumerate}
\end{prop}

\begin{proof} 1) We use induction on $n$. Suppose that $n=2$. Then
\begin{eqnarray*}
f(x\phi_{a_1} \phi_{a_2})  & = & f( (((x\lor \neg a_1)\land a_1) \lor \neg a_2) \land a_2 )   \\
& = & (((f(x)\lor \neg f(a_1))\land f(a_1)) \lor \neg f(a_2)) \land f(a_2) \\
& = &f(x)\phi_{f(a_1)} \phi_{f(a_2)}
\end{eqnarray*}
Suppose that the result holds for $m < n$. Then:
\begin{eqnarray*}
f(x\phi_{a_1} \ldots  \phi_{a_n})  & = & f((x\phi_{a_1} \ldots \phi_{a_{n-1}} \lor \neg a_n) \land a_n )   \\
& = & (f(x) \phi_{f(a_1)} \ldots  \phi_{f(a_{n-1})}  \lor \neg f(a_n)) \land f(a_n) \\
& = & f(x)\phi_{f(a_1)} \ldots \phi_{f(a_n)}
\end{eqnarray*}

2) Suppose that $f:L_1 \rightarrow L_2$ is a surjective
$IE_B$-homomorphism. If $\phi \in  {\mathbb S}_0(L_1)$ then $\phi =
\phi_{a_1}  \ldots \phi_{a_n}$ where $\phi_{a_i}$ are Sasaki
projections on $L_1$. We define the function $g: {\mathbb S}_0(L_1)
\rightarrow {\mathbb S}_0(L_2)$ such that $$g(\phi) = g(\phi_{a_1}
\ldots \phi_{a_n}) = \phi_{f(a_1)} \ldots \phi_{f(a_n)}
$$  We first prove that $g$ is well defined. Suppose that $\phi =
\phi_{a_1} \ldots \phi_{a_n} =  \phi_{c_1} \ldots \phi_{c_m}$. Let
$b \in L_2$. Since $f$ is a surjective function then there exists
$a\in L_1$ such that $f(a) = b$. Then by item 1,
\begin{eqnarray*}
b\phi_{f(c_1)} \ldots \phi_{f(c_m)}   & = & f(a) \phi_{f(c_1)} \ldots \phi_{f(c_m)}   \\
& = & f(a\phi_{c_1} \ldots \phi_{c_m}) \\
& = & f(a\phi_{a_1} \ldots \phi_{a_n}) \\
& = & f(a)\phi_{f(a_1)} \ldots \phi_{f(a_n)} \\
& = & b \phi_{f(a_1)}  \ldots  \phi_{f(a_n)}
\end{eqnarray*}
Thus $g(\phi_{a_1} \ldots  \phi_{a_n}) = g(\phi_{c_1} \ldots
\phi_{c_m})$ and  $g$ is well defined. Note that for each $a\in
L_1$, $g(\phi_a) = \phi_{f(a)}$ and then $g/_{L_1} = f$ identifying
$L_1$ with $P_c({\mathbb S}_0(L_1))$. The surjectivity of $g$
follows immediately from the surjectivity of $f$. By definition of
$g$, it is immediate that $g$ is a $\langle \circ, ^*, 0
\rangle$-homomorphism where $\psi \circ \phi = \psi \phi$. We prove
that $g$ preserves the operation $'$. Suppose that $\phi =
\phi_{a_1} \ldots  \phi_{a_n}$. By Theorem \ref{PRO2} $\phi' =
\phi_{\neg 1\phi} = \phi_{\neg 1\phi_{a_1} \ldots  \phi_{a_n}}$. By
item 1 we have that
\begin{eqnarray*}
g(\phi') & = & g(\phi_{\neg 1\phi_{a_1} \ldots  \phi_{a_n}})   \\
& = & \phi_{f(\neg 1\phi_{a_1} \ldots  \phi_{a_n}} ) \\
& = & \phi_{\neg f(1)\phi_{f(a_1)}  \ldots \phi_{f(a_n)}}\\
& = & \phi_{\neg 1 g(\phi)}\\
& = & g(\phi)'
\end{eqnarray*}
Thus, $g$ preserves the operation $'$. Now we prove that $g$
preserves $s$. By Proposition \ref{BS2}-3 $s(\phi) = s(\phi'')$.
Then there exists $a\in L_1$ such that $\phi'' = \phi_a$. By
Corollary \ref{BS4}, $g(s(\phi)) = g(s(\phi_a)) = g(\phi_{s(a)}) =
\phi_{f(s(a))} = \phi_{s(f(a))} = s(\phi_{f(a)}) = s(g (\phi_a)) =
s(g(\phi'')) = s(g(\phi)'') = s(g(\phi))$ and $g$
preserves the operation $s$. Hence $g$ is a surjective $IE^*_B$-homomorphism
such that, identifying $L_1$ with $P_c({\mathbb S}_0(L_1))$, $g/_{L_1} = f$ . We have
to prove that $g$ is unique. Suppose that there exists an $IE^*_B$-homomorphism
$h: {\mathbb S}_0(L_1) \rightarrow {\mathbb S}_0(L_2)$ such that $h/_{L_1} = f$.
Let $\phi = \phi_{a_1} \ldots  \phi_{a_n} \in  {\mathbb S}_0(L_1)$. Then
$ h(\phi) = h(\phi_{a_1} \ldots \phi_{a_n}) = h(\phi_{a_1}) \ldots (\phi_{a_n}) =  f(\phi_{a_1}) \ldots  f(\phi_{a_n})
= g(\phi_{a_1} \ldots \phi_{a_n} ) = g(\phi)$. Thus, $h = g$ and this proves the unicity of $g$. \\

3) To prove this item, we assume that $f$ is bijective and use the
function $g$ of item 2. Then we have to prove that $g$ is injective.
Suppose that $g(\phi) = g(\psi)$ where $\varphi, \psi \in {\mathbb
S}_0(L_1)$. Suppose that $\phi = \phi_{a_1} \ldots  \phi_{a_n}$  and
$\psi =\phi_{c_1}  \ldots  \phi_{c_m}$. By item 1, for each $x\in
L_1$ we have that:
\begin{eqnarray*}
f(x\phi_{a_1}  \ldots  \phi_{a_n}) & = & f(x) \phi_{f(a_1)}  \ldots  \phi_{f(a_n)}\\
& = & f(x)g(\phi_{a_1} \ldots \phi_{a_n})\\
& = & f(x)g(\phi) \\
& = & f(x)g(\psi)\\
& = & f(x)\phi_{f(c_1)} \ldots \phi_{f(c_m)}\\
& = & f(x\phi_{c_1} \ldots \phi_{c_m})
\end{eqnarray*}
Since $f$ is bijective, $\phi_{a_1} \ldots \phi_{a_n}(x) = \phi_{c_1}  \ldots  \phi_{c_m}(x)$ and then $\phi = \psi$. Thus $g$ is bijective.

\rule{5pt}{5pt}
\end{proof}

\begin{prop} \label{HOMBAER1}
Let $A$ be a sub $IE_B$-lattice of $L$. Then there exists a sub
$IE^*_B$-semigroup $S_A$ of ${\mathbb S}_0(L)$ such that $A$ is
$IE_B$-isomorphic to $P_c(S_A)$.

\end{prop}

\begin{proof}
Consider the set $$S_A = \bigcup_{n \in {\mathbb{N}}} \{
\phi_{a_1}\phi_{a_2}\ldots \phi_{a_n}: a_i\in A \}$$ where
$\phi_{a_i}$ are Sasaki projections on $L$. Note that in general
$S_A \not = {\mathbb S}_0(A)$ since the domain of Sasaki projections
$\phi_{a_i}$ is $L$ (and not $A$). In {\rm \cite[Proposition
10]{AD}} it is proved that $S_A$ is a sub Baer $^*$-semigroup of
${\mathbb S}_0(L)$ in which $A$ is $OML$-isomorphic to $P_c(S_A)$
and then, $A$ is $IE_B$-isomorphic to $P_c(S_A)$. Thus, by Corollary
\ref{BS5}, $S_A$ is a sub $IE^*_B$-semigroup  of ${\mathbb S}_0(L)$.
\rule{5pt}{5pt}
\end{proof}

\begin{prop} \label{HOMBAER3}
Let $S$ be an $IE^*_B$-semigroup and for each $a\in S$ we define the
function $\psi_a: P_c(S) \rightarrow P_c(S)$ such that $\psi_a(x) =
(x a)''$. Then

\begin{enumerate}

\item
If $a \in P_c(S)$ then $\psi_a = \phi_a$.

\item
$f: S \rightarrow S(P_c(S))$ such that $f(a) = \psi_a$ is an  $IE^*_B$-homomorphism.

\item
Let $S_0$ be the sub $IE^*_B$-semigroup of $S$ generated by
$P_c(S)$. If we consider the restriction $f/_{S_0}$ then
Imag($f/_{S_0}$)$ = {\mathbb S}_0(P_c(S))$.

\end{enumerate}
\end{prop}

\begin{proof}
1) Suppose that $a\in  P_c(S)$. By {\rm \cite[Lemma 37.10 ]{MM}}, $\psi_a(x) = (x a)'' = (x \lor \neg a) \land a = x\phi_a$. Hence $\psi_a = \phi_a$.

2)  In {\rm \cite[Proposition 7]{AD}} is proved that $f$ preserves
the operations  $\langle \cdot,^*,',0 \rangle$. Then we have to
prove that $f$ preserves the operation $s$. Note that, by item 1,
$f/_{P_c(S)}$ is the $IE_B$-isomorphism $a\mapsto \phi_a$. Then, by
Corollary \ref{BS4}, $f(s(a)) = s(f(a))$ for each $a\in P_c(S)$.
Taking into account that for each $a\in S$, $s(a) = s(a'')$ we have
that, $f(s(a)) = f(s(a'')) = s(f(a'')) = s(f(a)'') = s(f(a))$. Thus
$f$ is an  $IE^*_B$-homomorphism.

3) Suppose that $\varphi \in {\mathbb S}_0(P_c(S))$. Then $\varphi =
\phi_{a_1} \ldots  \phi_{a_n}$ for some $a_1, \ldots ,a_n$ in
$P_c(S)$. If we consider the element $a = a_1 a_2 \ldots  a_n$ then
$a\in S_0$ and, since $f$ is an $IE^*_B$-homomorphism, $f(a) = f(a_1
a_2 \ldots  a_n) = f(a_1) f(a_2) \ldots f(a_n) = \phi_{a_1}
\phi_{a_2} \ldots \phi_{a_n} = \varphi$. Imag($f/_{S_0}$)$ =
{\mathbb S}_0(P_c(S))$.

\rule{5pt}{5pt}
\end{proof}

\begin{prop} \label{PRODG0}
Let $(L_i)_{i\in I}$ be a family of $IE_B$-lattices. Then:

\begin{enumerate}

\item
If $\vec{a} = (a_i)_{i\in I} \in \prod_{i\in I} L_i$ then the Sasaki
projection $\phi_{\vec{a}}: \prod_{i\in I} L_i \rightarrow
\prod_{i\in I} L_i$ satisfies that for each $\vec{x} = (x_i)_{i\in
I}$, $\vec{x}\phi_{\vec{a}} = (x_i\phi_{a_i})_{i\in I}$.

\item
If  $\vec{a} = (a_i)_{i\in I}$ and $\vec{b} = (b_i)_{i\in I}$ are
elements in $\prod_{i\in I} L_i$ then for each $\vec{x} =
(x_i)_{i\in I}$, $\vec{x} \phi_{\vec{a}}  \phi_{\vec{b}} =
(x_i\phi_{a_i} \phi_{b_i})_{i\in I}$.

\item
 ${\mathbb S}_0(\prod_{i\in I} L_i)$ is  $IE^*_B$-isomorphic to $\prod_{i\in I} {\mathbb S}_0 (L_i)$

\end{enumerate}

\end{prop}

\begin{proof}
1) Let $\vec{a} = (a_i)_{i\in I} \in \prod_{i\in I} L_i$. Then
\begin{eqnarray*}
\vec{x} \phi_{\vec{a}}  & = & ( (x_i)_{i\in I} \lor \neg (a_i)_{i\in I} ) \land (a_i)_{i\in I}  \\
& = & (( x_i \lor \neg a_i ) \land a_i)_{i\in I}\\
& = & (x_i\phi_{a_i})_{i\in I}
\end{eqnarray*}

2) Let  $\vec{a} = (a_i)_{i\in I}$ and $\vec{b} = (b_i)_{i\in I}$ be two elements in $\prod_{i\in I} L_i$ and $\vec{x} = (x_i)_{i\in I}$. Then, by item 1, we have that:
\begin{eqnarray*}
\vec{x} \phi_{\vec{a}}  \phi_{\vec{b}}  & = & (\vec{x} \phi_{\vec{a}})  \phi_{\vec{b}}   \\
& = & ((x_i\phi_{a_i})_{i\in I}) \phi_{\vec{b}}\\
& = &  (x_i\phi_{a_i} \phi_{b_i})_{i\in I}
\end{eqnarray*}

3) Follows from item 2.

\rule{5pt}{5pt}
\end{proof}

\begin{prop} \label{DIRECTINDB}
Let $L$ be an  $IE_B$-lattice. Then, $L$ is directly indecomposable
iff ${\mathbb S}_0(L)$ is directly indecomposable
\end{prop}

\begin{proof}
Suppose that ${\mathbb S}_0(L)$ admits a non trivial decomposition
in direct products of $IE^*_B$-semigroups i.e. ${\mathbb S}_0(L) =
\prod_{i\in I} S_i$. Then, by Proposition \ref{PRO11}, we can see
that $L \approx_{IE_B} P_c({\mathbb S}_0(L)) \approx_{IE_B}
P_c(\prod_{i\in I} S_i) \approx_{IE_B} \prod_{i\in I} P_c(S_i)$.
Thus $L$ admits a non trivial decomposition in direct products of
$IE_B$-lattices.

Suppose that $L$ admits a non trivial decomposition in direct
products of $IE_B$-lattices i.e. $L = \prod_{i\in I} L_i$. Then, by
Proposition \ref{PRODG0}-3, ${\mathbb S}_0(L) = {\mathbb
S}_0(\prod_{i\in I} L_i) \approx_{IE^*_B} \prod_{i\in I} {\mathbb
S}_0(L_i)$. Thus ${\mathbb S}_0(L)$ admits a non trivial
decomposition in direct products of $IE^*_B$-semigroups.
\rule{5pt}{5pt}
\end{proof}

Let ${\mathcal A}_I$ be a variety of $IE_B$-lattices. We denote by
${\mathcal G}^*({\mathcal A}_I)$ the sub variety of ${\mathcal
A}_I^*$ generated by the class $\{{\mathbb S}_0(L): L \in {\mathcal
A}_I \}$. More precisely,
$${\mathcal G}^*({\mathcal A}_I) = {\mathcal V}(\{{\mathbb S}_0(L): L \in {\mathcal
A}_I \}) $$

We also introduce the following subclass of ${\mathcal
G}^*({\mathcal A}_I)$  $${\mathcal G}_D^*({\mathcal A}_I) =
\{{\mathbb S}_0(L): L \in {\mathcal D}({\mathcal A}_I) \} $$ where
${\mathcal D}({\mathcal A}_I)$ is the class of the direct
indecomposable algebras of ${\mathcal A}_I$. By Proposition
\ref{DIRECTINDB}, we can see that ${\mathcal G}_D^*({\mathcal A}_I)$
is a subclass of the direct indecomposable algebras of  ${\mathcal
A}_I^*$.

\begin{theo}\label{COMPLETENESSEXP}
Let ${\mathcal A}_I$ be a variety of $IE_B$-lattices. Then
$${\mathcal G}^*({\mathcal A}_I) \models t= r \hspace{0.3cm}
\mathrm{iff} \hspace{0.3cm} {\mathcal G}_D^*({\mathcal A}_I) \models
t= r
$$
\end{theo}

\begin{proof}
As regards to the non-trivial direction assume that $${\mathcal
G}_D^*({\mathcal A}_I)  \models t(x_1, \ldots,x_n) = r(x_1,
\ldots,x_n)$$ Let ${\mathbb S}_0(L) \in {\mathcal G}^*({\mathcal
A}_I)$. By the subdirect representation theorem, there exists an
$IE_B$-lattice embedding $\iota: L \hookrightarrow \prod_{i\in I}
L_i $ where $(L_i)_{i\in I}$ is a family of subdirectly irreducible
algebras in ${\mathcal A}_I$. Therefore, $L_i \in  {\mathcal
D}({\mathcal A}_I) $ and $ {\mathbb S}_0(L_i) \in  {\mathcal
G}_D^*({\mathcal A}_I)$ for each $i\in I$. By Proposition
\ref{HOMBAER1}, there exists an $IE^*_B$-semigroup embedding
$\iota_F: F \hookrightarrow {\mathbb S}_0(\prod_{i\in I} L_i)$ where
$L$ is $IE_B$-isomorphic to $P_c(F)$. By Proposition \ref{PRODG0},
we can assume that the $IE^*_B$-semigroup embedding $\iota_F$ is of
the form $\iota_F: F \hookrightarrow \prod_{i\in I} {\mathbb
S}_0(L_i)$. By Proposition \ref{HOMBAER3}, if we consider the sub
$IE^*_B$-semigroup $F_0$ of $F$ generated by $P_c(F)$ then, there
exists a surjective $IE^*_B$-homomorphisms $f:F_0 \rightarrow
{\mathbb S}_0(L)$. The following diagram provides some intuition:

\begin{center}
\unitlength=1mm
\begin{picture}(90,20)(0,0)
\put(3,10){\vector(0,-2){5}}
\put(20,16){\makebox(0,0){$F_0 \hookrightarrow F \stackrel{\iota_F}{\hookrightarrow} \prod_{i\in I} {\mathbb S}_0(L_i)$}}
\put(3,-1){\makebox(0,0){${\mathbb S}_0(L)$}}
\put(2,8){\makebox(-6,0){$f$}}
\end{picture}
\end{center}

\vspace{0.3cm}

Since $F_0$ can be embedded into a direct product $\prod_{i\in I}
{\mathbb S}_0(L_i)$, where $ {\mathbb S}_0(L_i) \in {\mathcal
G}_D^*({\mathcal A}_I)$ for each $i\in I$,  by hypothesis, we have
that:
$$F_0 \models t(x_1, \ldots,x_n) = r(x_1, \ldots,x_n)$$ Let $\vec{a}
= (a_1 \ldots a_n)$ be a sequence in ${\mathbb S}_0(L)$. Since $f$
is surjective, there exists a sequence $\vec{m} = (m_1,\ldots, m_n)$
in $F_0$ such that $f(\vec{m}) = (f(m_1),\ldots, f(m_n)) = \vec{a}$.
Since $t^{F_0}(\vec{m}) = r^{F_0}(\vec{m})$ then $t^{{\mathbb
S}_0(L)}(\vec{a}) = f(t^{F_0}(\vec{m})) = f(r^{F_0}(\vec{m})) =
t^{{\mathbb S}_0(L)}(\vec{a})$. Hence ${\mathbb S}_0(L) \models
t(x_1, \ldots,x_n) = r(x_1, \ldots,x_n)$ and the equation holds in
${\mathcal G}^*({\mathcal A}_I)$. \rule{5pt}{5pt}
\end{proof}

Even though the study of equations in Exp(${\mathcal A}_I$) is quite
treatable from the result obtained in Theorem \ref{COMPLETENESSEXP},
we do not have in general a full description of the equational
system that defines the variety Exp(${\mathcal A}_I$). The following
corollary provides an interesting property about Exp(${\mathcal
A}_I$).

\begin{coro}
Let ${\mathcal A}_I$ be a variety of $IE_B$-lattices. Then
$${\mathcal G}^*({\mathcal A}_I) \models s(x)\cdot y = y\cdot s(x).$$

\end{coro}

\begin{proof}
Let $S$ be an algebra in ${\mathcal G}_D^*({\mathcal A}_I)$. Then
for each $x\in S$, $s(x) \in \{0,1\}$ and $s(x)\cdot y = y\cdot
s(x)$. Hence by Theorem \ref{COMPLETENESSEXP}, ${\mathcal
G}^*({\mathcal A}_I) \models s(x)\cdot y = y\cdot s(x)$.
\rule{5pt}{5pt}
\end{proof}\\

Let ${\mathcal A}_I$ be a variety of $IE_B$-lattices. Note that the
assignment $L \mapsto {\mathbb S}_0(L)$ defines a class operator of
the form $${\mathbb S}_0:{\mathcal A}_I \rightarrow {\mathcal
G}^*({\mathcal A}_I) \subseteq {\mathcal A}^*_I$$ Taking into
account Definition \ref{problem3}, by  Proposition \ref{DIRECTINDB}
and Theorem \ref{COMPLETENESSEXP} we can establish the following
result:

\begin{theo}\label{DETER}
Let  ${\mathcal A}$ be a class of $E_B$-lattices. Suppose that the
variety  ${\mathcal A}_I$ of  $IE_B$-lattices equationally
characterizes the class ${\mathcal A}$. Then the class  ${\mathcal
A}$ determines the equational theory of ${\mathcal G}^*({\mathcal
A}_I)$. \rule{5pt}{5pt}
\end{theo}

\noindent This last theorem provides a solution to the problem posed in Section \ref{problem1}.

\section{Final remarks}

We have developed an algebraic framework that allows us to extend
families of two-valued states on orthomodular lattices to Baer
$^*$-semigroups. To do so, we have explicitly enriched this variety
with a unary operation that  captures the concept of two-valued
states on Baer $^*$-semigroups as an equational theory. Moreover, a
decidable method to find the equational system is given. We have
also applied this general approach to study the full class of
two-valued states and the subclass of Jauch-Piron two-valued states
on Baer $^*$-semigroups.

\section*{Acknowledgments}

The authors wish to thank an anonymous referee for his/her careful
reading of our manuscript and useful comments on an earlier draft.
His/her remarks have substantially improved our paper. This work was
partially supported by the following grants: PIP 112-201101-00636,
Ubacyt 2011/2014 635, FWO project G.0405.08 and FWO-research
community W0.030.06. CONICET RES. 4541-12 (2013-2014).\\

\end{document}